\documentclass[twoside,leqno,twocolumn]{article}

\usepackage[letterpaper]{geometry}

\usepackage{siamproceedings}
\usepackage{rotating}
\newcommand{\cpp}{\textsf{C}\texttt{++}\xspace}
\usepackage[T1]{fontenc}
\usepackage{amsfonts}
\usepackage{graphicx}
\usepackage{epstopdf}
\usepackage{enumitem}
\usepackage{hyperref}
\usepackage[noend]{algorithmic}
\usepackage{thm-restate}
\usepackage{tabularray}
\usepackage{booktabs}
\usepackage{enumitem}
\usepackage{xspace}
\usepackage{numprint}
\usepackage{siunitx}
\UseTblrLibrary{siunitx}
\UseTblrLibrary{functional}

        \ExplSyntaxOn
\cs_new:Npn \expandableinput #1
  { \use:c { @@input } { \file_full_name:n {#1} } }
\AddToHook{env/tabular/begin}
  { \cs_set_eq:NN \input \expandableinput }
\ExplSyntaxOff
\usepackage[skip = 0.2pt]{caption}
\ifpdf
  \DeclareGraphicsExtensions{.eps,.pdf,.png,.jpg}
\else
  \DeclareGraphicsExtensions{.eps}
\fi

\newcommand{\competitorCHTree}{\texttt{CH\_Tree}\xspace}
\newcommand{\competitorCQTree}{\texttt{CQ\_Tree}\xspace}
\newcommand{\competitorBTree}{\texttt{B\_Tree}\xspace}
\newcommand{\competitorVector}{\texttt{Vector\_Insert}\xspace}
\newcommand{\competitorAVL}{\texttt{AVL\_Tree}\xspace}
\newcommand{\competitorTerrible}{\texttt{Naïve \competitorVector}\xspace}

\newcommand{\bucketNameBase}{\texttt{Logarithmic}\xspace}

\newcommand{\logarithmic}{\texttt{Logarithmic\_Method}\xspace}
\newcommand{\bucketLinearGrowSuffix}{\texttt{Linear}\xspace}
\newcommand{\bucketHullSuffix}{\texttt{Hull}\xspace}
\newcommand{\bucketBtreeSuffix}{\texttt{BTree}\xspace}
\newcommand{\bucketLinearGrow}{\bucketNameBase (\bucketLinearGrowSuffix)\xspace}
\newcommand{\bucketHull}{\bucketNameBase (\bucketHullSuffix)\xspace}
\newcommand{\bucketBtree}{\bucketNameBase (\bucketBtreeSuffix)\xspace}

\newsiamremark{remark}{Remark}
\newsiamremark{hypothesis}{Hypothesis}
\crefname{hypothesis}{Hypothesis}{Hypotheses}
\newsiamthm{claim}{Claim}

\usepackage{amsopn}

\begin{document}

\newcommand\relatedversion{}
\newtheorem{observation}{Observation}

\title{\Large Practical Insertion-Only Convex Hull}%

\author{Ivor van der Hoog\thanks{IT University of Copenhagen, Denmark}\and Henrik Reinstädtler\thanks{ Heidelberg University, Germany} \and Eva Rotenberg\footnotemark[1]}
\date{}

\maketitle

\fancyfoot[R]{\scriptsize{Copyright \textcopyright\ 2025\\
Ivor van der Hoog, Henrik Reinstädtler and Eva Rotenberg}}

\begin{abstract} 
Convex hull data structures are fundamental in computational geometry. We study insertion-only data structures for convex hulls of a planar point set $P$ of size $n$, supporting various containment and intersection queries. When $P$ is sorted by $x$- or $y$-coordinate, convex hulls can be constructed in linear time using classical algorithms such as Graham scan. In the fully dynamic setting, the algorithm by Overmars and van Leeuwen~\cite{overmars1980dynamically} maintains the convex hull under insertions and deletions in $O(\log^2 n)$ time per update, supports queries in time logarithmic in the size of the convex hull, and uses $O(n)$ space. An open-source implementation of their method is available.

We investigate a variety of methods tailored to the insertion-only setting. We explore a broad selection of trade-offs involving robustness, memory access patterns, and space usage, providing an extensive evaluation of both existing and novel techniques.

We observe that all logarithmic-time methods rely on pointer-based tree structures, which suffer in practice due to poor memory locality. Motivated by this, we develop a vector-based solution inspired by Overmars' logarithmic method~\cite{overmars1983design}. Our structure has worse asymptotic bounds, supporting queries in $O(\log^2 n)$ time, but stores data in $O(\log n)$ contiguous vectors, greatly improving cache performance.

Through empirical evaluation on real-world and synthetic data sets, we uncover surprising trends. Let $h$ denote the size of the convex hull. We show that a naïve $O(h)$ insertion-only algorithm based on Graham scan consistently outperforms both theoretical and practical state-of-the-art methods under realistic workloads, even on data sets with rather large convex hulls.  While tree-based methods with $O(\log h)$ update times offer solid theoretical guarantees, they are never optimal in practice. In contrast, our vector-based logarithmic method, despite its theoretically inferior bounds, is highly competitive across all tested scenarios. It is optimal whenever the convex hull becomes large. 
\end{abstract}
\newpage
\section*{Acknowledgements:}
{\it Ivor van der Hoog} and  {\it Eva Rotenberg} are grateful to the VILLUM Foundation for supporting this research via Eva Rotenberg's Young Investigator grant (VIL37507) ``Efficient Recomputations for Changeful Problems''. Part of this work took place while these authors were affiliated with the Technical University of Denmark.
\\
{\it Ivor van der Hoog} This project has received funding from the European Union's Horizon 2020 research and innovation programme under the Marie Sk\l{}odowska-Curie grant agreement No 899987. \\
{\it Henrik Renstädtler} acknowledges support by DFG grant \hbox{SCHU 2567/8-1}.
\clearpage
\section{\texorpdfstring{Introduction.}{Introduction.}}
Let $P$ be a planar point set of size $n$. The convex hull of $P$ is the smallest~convex area enclosing $P$.
Convex hulls are among the most~studied objects in computational geometry~\cite{Avis1995How, brodal2002dynamic}, with applications in clustering~\cite{gao2017trajectory,liparulo2015fuzzy,sander1998density}, shape analysis~\cite{Furukawa1995, Wang2007Supervised}, data pruning~\cite{giorginis2022fast, khosravani2016convex, margineantu1997pruning, ostrouchov2005fastmap}, selection queries~\cite{ihm2014approximate, mouratidis2017geometric}, and road network analysis~\cite{gao2017trajectory,liu2019fast, yan2011efficient}.

\subsection{Related work.}
We denote by $CH(P)$ the convex hull edges in cyclic order and $h  := |CH(P)|$. 
Dynamic convex hulls have been studied for several decades. 
A \emph{fully dynamic explicit data structure} maintains $CH(P)$ in sorted order. The classic algorithm by Overmars and van Leeuwen~\cite{overmars1980dynamically} explicitly maintains $CH(P)$ subject to point insertions and deletions in $O(\log^2 n)$ time. This remains the most efficient \emph{explicit} method to date.
Other works progressed by restricting the problem statement. Brewer et al.~\cite{brewer2024Dynamic} restrict $P$ to be an ordered set that form the vertices of simple path: \cite{brewer2024Dynamic} allows updates that append points (without introducing a crossing), and the reverse pop-operation, both in $O(\log n)$ time. Wang~\cite{Wang2023Dynamic} improves this to $O(\log h)$ in the case where $P$ remains $x$-monotone.

\emph{Implicit hulls}, in contrast, maintain auxiliary structures on $P$ to answer queries~\cite{chan2011three}, such as e.g. containment queries. Existing techniques support only \emph{non-decomposable} queries (which are defined in Section~\ref{sec:prelim}) in logarithmic time. 
Friedman et al.~\cite{Friedman1996Efficiently} were the first to consider implicit structures.
They restrict $P$ to a simple path and support appending and removing points along a path in amortised $O(\log n)$ and $O(1)$ time respectively.
Chan's~\cite{chan2001dynamic}  dynamic structure has $O(\log^{1+\varepsilon} n)$ amortised update time and $O(\log^{3/2} n)$ query time, later improved to $O(\log^{1+\varepsilon} n)$ expected time~\cite{chan2011three}.
Brodal and Jacob~\cite{Brodal2000Dynamic} obtain amortised $O(\log n \log \log n)$ update and $O(\log n)$ query time. 
This was improved to $O(\log n)$-time~\cite{brodal2002dynamic}.

\paragraph{Insertion-only.} 
Preparata~\cite{Preparata1979}, 
introduced the first explicit insertion-only convex hull algorithm. This solution has a $O(\log n)$ amortised update time. A simpler folklore adaption of Graham scan~\cite{graham1972efficient}, that we detail in our preliminaries, achieves $O(\log h)$ time.

Despite their broad applicability, dynamic convex hull \emph{implementations} are scarce. The Computational Geometry Algorithms Library (CGAL) offers only a dynamic implementation of Delaunay triangulations, which offer poor linear-time update performance. Chi et al.~\cite{chi2013distribution} provide a Java implementation of the Overmars–van Leeuwen algorithm, though it lacks query support, robustness, and efficiency~\cite{Gaede2024Simple}. 
The state-of-the-art is by Gæde et al.~\cite{Gaede2024Simple}, who use CGAL's exact computation types to robustly implement the $O(\log^2 n)$ algorithm from~\cite{overmars1980dynamically}. Implicit methods \cite{brodal2002dynamic, chan2011three, chan2001dynamic} remain unimplemented. These rely on involved cuttings that require significant simplifications to become implementable. 
\subsection{Contribution.}
We study practically efficient methods for maintaining a convex hull in an insertion-only setting. Starting from the folklore variant of Graham scan, we explore and evaluate several implementation design choices and efficiency parameters such~as:
\begin{itemize}[noitemsep]
    \item Robustness considerations, 
    \item Memory access patterns, and
    \item Practical space usage. 
\end{itemize}
\noindent
We thereby provide an elaborate, empirically tested, overview of insertion-only convex hull techniques.

Our primary technical contribution is an insertion-only convex hull data structure based on Overmars' logarithmic method~\cite{overmars1983design}.
We provide several insights to achieve amortised $O(\log n)$ update time and  $O(\log^2 n)$ query time. 
 Our structure partitions the input $P$ into subsets $P_1, \ldots, P_k$, maintaining $CH(P_i)$ in contiguous memory, improving cache efficiency.
 
We perform an extensive experimental evaluation, comparing our method against both classical fully dynamic approaches and other insertion-only techniques. We find that all insertion-only methods vastly outperform fully dynamic algorithms such as~\cite{Gaede2024Simple}, often by orders of magnitude. While tree-based insertion-only methods have favourable asymptotic bounds, we show that in practice they underperform. In contrast, our logarithmic method is either competitive or optimal across all tested scenarios. A surprising outcome of our study is that a naïve linear-time approach based on Graham scan outperforms all alternatives on real-world data, even if the convex hull is moderately large. 

Finally, our structure provides a pathway toward practical, fully dynamic convex hull maintenance. We show that given any deletion-only subroutine, our technique can be extended to support deletions also.

\section{\texorpdfstring{Preliminaries. }{Preliminaries.}}\label{sec:prelim}
Our input is a planar point set $P$. We denote by $n$ the size of $P$ at the time of an insertion, deletion, or query. We assume that $n > 1$. The \emph{convex hull} of $P$ is the minimal convex area enclosing all points of $P$, we denote by $CH(P)$ its edges in cyclical order. For any edge $(a, b)$, we denote its supporting line by $line(a, b)$.

Let $\{ \ell, \rho, \tau, \beta \}$ denote the leftmost, rightmost, topmost, and bottommost points of $P$, respectively. These extremal points must lie on $CH(P)$. It is standard practice~\cite{Gaede2024Simple, overmars1980dynamically, Wang2023Dynamic} to maintain only the \emph{upper convex hull} $CH^+(P)$, defined as the subchain of $CH(P)$ from $\ell$ to $\rho$ via $\tau$. Maintaining $CH^+(P)$ and $CH^+(P')$, where $P'$ is $P$ with mirrored $y$-coordinates, suffices to support all convex hull queries.
For algorithmic simplicity when describing our techniques, we deviate from this convention and instead maintain a \emph{quarter hull}:

\begin{definition}
We define the \emph{quarter hull} $CH^\ulcorner(P)$ as the subsequence of $CH^+(P)$ from $\ell$ to $\tau$. \\
We define $CH^\urcorner(P)$ analogously from $\tau$ to $\rho$.
\end{definition}

\begin{observation}
\label{obs:decreasing}
The edges of $CH^\ulcorner(P)$ have positive slope, and consecutive slopes are decreasing.
\end{observation}

\noindent
Our implementations maintain the full upper hull $CH^+(P)$ by maintaining $CH^\ulcorner(P)$ and $CH^\urcorner(P)$ separately. We revisit the static, insertion-only, and fully dynamic algorithms that construct or maintain $CH^\ulcorner(P)$.

\paragraph{Queries.}
Chan identifies six convex hull queries~\cite{chan2011three}:
\begin{enumerate}[noitemsep]
    \item Finding the extreme point of $P$ in a given direction,
    \item Deciding whether a query line intersects $CH(P)$,
    \item Finding the hull vertices tangent to a query line,
    \item Deciding whether a point $q$ lies inside the area bounded by $CH(P)$,
    \item Finding the intersection points with a query line,
    \item Finding the intersection between two convex hulls.
\end{enumerate}

The first three are \emph{decomposable queries}: if $P$ is partitioned into $P_1$ and $P_2$, then the answer for $P$ can be computed in constant time from the answers for $CH(P_1)$ and $CH(P_2)$. These queries are therefore easier to support and are handled efficiently by implicit convex hull data structures~\cite{Brodal2000Dynamic, brodal2002dynamic, chan2011three,chan2001dynamic, Friedman1996Efficiently}.
The latter three are \emph{non-decomposable} and are typically not supported by implicit structures (Chan~\cite{chan2011three} supports these in $O(\log^{3/2} n)$ time). For general applicability, and to match CGAL and~\cite{chi2013distribution, Gaede2024Simple}, 
these non-decomposable queries will be our primary focus (see also Appendix~\ref{app:queries}).

\paragraph{Static construction.}
Constructing $CH^\ulcorner(P)$ has a natural $\Omega(n \log n)$ lower bound via a reduction from sorting. Let $P = (p_1, \ldots, p_n)$ be sorted by $x$-coordinate. Then Graham scan~\cite{graham1972efficient} constructs $CH^\ulcorner(P)$ in linear time by computing a sequence of segments of decreasing positive slope, starting from $\ell$ (see Alg.~\ref{alg:grahamscan} and~\ref{alg:Scan}). These can be adapted to compute $CH^\urcorner(P)$ instead. 

\begin{algorithm}
    \caption{\texttt{Graham}($x$-sorted point set $P$)}
    \label{alg:grahamscan}
    \begin{algorithmic}
    \STATE $CH^\ulcorner \gets \{ (p_1, p_2) \}$
    \FOR{int $k \in [2, n]$ with $p_k$ left of  $\tau$}
        \STATE \texttt{Scan}($CH^\ulcorner$, $p_k$)
    \ENDFOR
    \RETURN $CH^\ulcorner$
    \end{algorithmic}
\end{algorithm}

\begin{algorithm}
    \caption{\texttt{Scan}(convex sequence of edges with positive slope $CH^\ulcorner$, point $q$ right of $CH^\ulcorner$)}
    \label{alg:Scan}
    \begin{algorithmic}
    \STATE $(u, v) \gets CH^\ulcorner$.last
    \WHILE{$\text{slope}({line}(u, v)) < \text{slope}({line}(v, q))$}
        \STATE $CH^\ulcorner$.remove($CH^\ulcorner$.last)
        \STATE $(u, v) \gets CH^\ulcorner$.last
    \ENDWHILE
    \STATE $CH^\ulcorner$.append($(v, q)$)
    \RETURN $CH^\ulcorner$
    \end{algorithmic}
\end{algorithm}

\noindent
In the while-loop of \texttt{Scan}, each iteration either permanently removes a vertex from $CH^\ulcorner$ or terminates. Thus, it Graham scan runs in $O(n)$ total time.

\paragraph{Insertion-only Graham.}
It is folklore that \texttt{Scan} enables insertion-only convex hulls (see also Algorithm~\ref{alg:insertiononly}):
First test whether the insertion point $q$ lies above the hull. 
If so, split $CH^\ulcorner(P)$ into two sequences: $S_1$, preceding $q$ in $x$-coordinate, and $S_2$, succeeding $q$.
Then use \texttt{Scan} to add $q$ to both sequences. For this, we implicitly interpret $S_2$ in reverse order and with $x$-coordinates mirrored around $q$.
The new hull $CH^\ulcorner(P \cup \{q\})$ is the concatenation of the two updated subsequences.

\begin{algorithm}
    \caption{\texttt{Insert}($CH^\ulcorner(P)$, $q$ above $CH^\ulcorner(P)$)}
    \label{alg:insertiononly}
    \begin{algorithmic}
    \STATE $S_1 := \{v \in CH^\ulcorner(P), \textnormal{ left of } q \}$ 
    \STATE $S_2 := \{v \in CH^\ulcorner(P), \textnormal{ right of } q \}$, implicitly reversed
    \STATE $S_1 \gets$ \texttt{Scan}($S_1$, $q$)
    \STATE $S_2 \gets$ \texttt{Scan}($S_2$, $q$)
    \RETURN $S_1 \cup $ non-reversed($S_2$)
    \end{algorithmic}
\end{algorithm}

\begin{theorem}[Folklore]
\label{thm:folklore}
Let $P$ be an insertion-only point set and denote by $h$ the size of $CH^\ulcorner(P)$ at update time. Algorithm~\ref{alg:insertiononly} maintains $CH^\ulcorner(P)$ using $O(h)$ space with $O(\log h)$ amortised update time.
\end{theorem}

\begin{proof}
    Searching $CH\ulcorner(P)$ takes $O(\log h)$ time. 
    Whenever \texttt{Scan}($CH^\ulcorner$, $q$) removes an edge $(u, v)$, the vertex $v$ lies under the line $line(u, q)$. Thus, insertion-only, the vertex $v$ can never again appear on $CH(P)$ and this scan takes amortised constant time. 
\end{proof}

\paragraph{Fully dynamic convex hull.}
We briefly describe for completeness the fully dynamic algorithm by Overmars and van Leeuwen to maintain $CH^+(P)$. 
Let $P_1$ and $P_2$ be two point sets, separated by a vertical line and let $P = P_1 \cup P_2$.
 Then $CH^+(P)$ consists of:
 \begin{itemize}[noitemsep]
     \item A contiguous subsequence of $CH^+(P_1)$,
     \item A contiguous subsequence of $CH^+(P_2)$
     \item The edge $e$ tangent to $CH^+(P_1)$ and $CH^+(P_2)$.
 \end{itemize}
Given $CH^+(P_1)$ and $CH^+(P_2)$, these three components can be computed in $O(\log n)$ time. Their data structure stores $P$ in a balanced binary tree $T$, sorted by $x$-coordinate. Each internal node considers the point sets $P_1$ and $P_2$ in its two child subtrees. Within the node, it stores these three above components. By combining these components, one can maintain at the root of the tree a binary tree that stores $CH^+(P)$ in cyclical order.

Each update affects one root-to-leaf path in $T$ of length $O(\log n)$. Updating all edges $e$ along this path takes $O(\log^2 n)$ time. The structure uses these edges $e$ to maintain $CH^+(P)$ at no asymptotic overhead.

\section{Engineering Insertion-only Convex Hulls.}
\label{sec:engineering}

Algorithm~\ref{alg:grahamscan} provides a linear-time static procedure to construct $CH^\ulcorner(P)$, while Algorithm~\ref{alg:insertiononly} implements an insertion-only approach with amortised $O(\log h)$ update time. Implementing these high-level algorithms gives rise to $3$ notable design decisions, which we detail below.

\paragraph{1. Storing $CH^\ulcorner(P)$ in a vector or pointer structure.}
Storing $CH^\ulcorner(P)$ in a dynamic vector optimises for query performance across all convex hull query types.
A vector ensures contiguous memory layout, which minimises cache misses and memory access overhead during traversal or search.
It is also memory-optimal, since no auxiliary pointers or node structures are required.

However, in an insertion-only setting inserting into vectors incurs additional costs.
Inserting a point $q$ at an intermediate location in $CH^\ulcorner(P)$ (i.e., not at the end) may require us to shift $\Theta(n)$ other elements. This increases the time complexity of Algorithm~\ref{alg:insertiononly} to $\Theta(n)$.

Pointer-based alternatives, such as balanced binary search trees, enable $O(\log h)$ searches and point insertions.  
However, they increase memory usage and lead to less efficient memory access patterns, which can hurt performance in practice. In Section~\ref{sec:implementations}, we detail several tree-based implementations.

\paragraph{2. Finger search.}
Algorithm~\ref{alg:Scan} performs a linear search through a sequence of edges $CH^\ulcorner$ via a while-loop.  
Given $q$, the loop identifies the maximal index $i$ such that the slope of edge $(u, v) = CH^\ulcorner[i]$ is less than that of $(v, q)$.  
If $|CH^\ulcorner| = m$, this linear scan takes $O(m - i)$ time. Theorem~\ref{thm:folklore} guarantees that this is amortised $O(1)$ over all insertions.
A faster practical method exists, though it yields no asymptotic gain:

\begin{observation}
\label{obs:binary}
For a convex sequence of $m$ edges $CH^\ulcorner$ and a point $q$, there exists a unique $i \in [m]$ s.t.:
\begin{itemize}[noitemsep]
    \item For all $j < i$, the slope of edge $CH^\ulcorner[j] = (u, v)$ satisfies $\text{slope}(\text{line}(u, v)) \geq \text{slope}(\text{line}(v, q))$.
    \item For all $j \geq i$, $\text{slope}(\text{line}(u, v)) < \text{slope}(\text{line}(v, q))$.
\end{itemize}
\end{observation}

\noindent
Observation~\ref{obs:binary} allows us to locate $i$ using binary search in $O(\log m)$ time. However, $O(\log m)$ may dominate $O(m - i)$. \emph{Finger search}~\cite{brodal2018finger} achieves $O(\log (m - i))$ time to identify $i$ instead.
Once $i$ is found, a sequence of $O(m - i)$ elements must still be removed from $CH^\ulcorner(P)$, as shown in Algorithms~\ref{alg:quickscan} and \ref{alg:quickinsert}.

\begin{algorithm}[H]
    \caption{\texttt{QuickScan}($CH^\ulcorner$, $q$, $m = |CH^\ulcorner|$)}
    \label{alg:quickscan}
    \begin{algorithmic}
    \STATE $f \gets$ $CH^\ulcorner$.last
    \STATE find, from $f$, $i :=$ minimum $j \in [m]$ where for $(u, v) = CH^\ulcorner[j]$, slope($line(u, v)$) $<$ slope($line(v,q)$)
    \RETURN $i$
    \end{algorithmic}
\end{algorithm}
\newpage

\begin{algorithm}[H]
    \caption{\texttt{QuickIns}($CH^\ulcorner(P)$, $q$ above $CH^\ulcorner(P)$)}
    \label{alg:quickinsert}
    \begin{algorithmic}
    \STATE $S_1 \gets \{v \in CH^\ulcorner(P) \mid \text{left of } q\}$
    \STATE $S_2 \gets \{v \in CH^\ulcorner(P) \mid \text{right of } q\}$, implicitly reversed
    \STATE $i \gets$ \texttt{QuickScan}($S_1$, $q$)
    \STATE $j \gets$ \texttt{QuickScan}($S_2$, $q$)
    \STATE Delete from $S_1$ its suffix starting at $S_1[i]$
    \STATE Delete from $S_2$ its prefix up to $S_2[j]$
    \RETURN $S_1 \cup \text{reverse}(S_2)$
    \end{algorithmic}
\end{algorithm}

\paragraph{3. Balanced deletion.} 
Given the index $i$ from Observation~\ref{obs:binary}, we must delete a contiguous sequence of edges from $CH^\ulcorner(P)$.  
During static construction (Algorithm~\ref{alg:grahamscan}), this sequence is a suffix of the current set of edges.  
For vectors, suffix deletion takes constant time by simply reducing the vector's size.
In a dynamic insertion-only setting, however, the deleted interval is no longer a suffix (see Algorithm~\ref{alg:quickinsert}).  
Any vector-based solution therefore incurs $\Omega(h)$ update time in the worst case.  
Tree-based solutions also face challenges: if $CH^\ulcorner(P)$ is stored in a balanced tree with $h$ leaves, then naively deleting $k$ elements requires $O(k \log h)$ time.  
Although the amortised cost remains $O(\log n)$, the constant factors are considerable.
In a balanced tree, a contiguous interval of $k$ leaves spans only $O(\log k)$ disjoint subtrees.  
These subtrees can be deleted individually, and the structure rebalanced in $O(\log k \log h)$ time.  
We refer to this approach as \emph{balanced deletion}.  

\section{\texorpdfstring{Implementations.}{Implementations.}}
\label{sec:implementations}

Given our considerations in Section~\ref{sec:engineering}, we implement our static Algorithm~\ref{alg:grahamscan} by storing $CH^\ulcorner(P)$ in a vector.
We use the \texttt{QuickIns} subroutine that uses \texttt{Quickscan}. 
We also provide $3$ implementations of Algorithm~\ref{alg:insertiononly}:

\begin{itemize}[noitemsep]
    \item \competitorVector maintains $CH^\ulcorner(P)$ in an ordered vector. Updates use Algorithm~\ref{alg:quickinsert}. 
    If $q \in CH(P \cup \{ q \})$ then updates can take $\Theta(h)$ time since we need to shift all vector elements that succeed $q$.
    \item \competitorAVL maintains $CH^\ulcorner(P)$ in an AVL-tree with balanced deletions.
        \item \competitorBTree maintains $CH^\ulcorner(P)$ in a B-tree. A $B$-tree guarantees better memory access patterns during updates and queries.  We use a fork of Google's \texttt{cpp-btree} implementation~\cite{JGRennison_cpp-btree} and augment the $B$-tree to support balanced deletions. 
\end{itemize}

\noindent
Our balanced binary trees do not use finger search (we use binary search in Algorithm~\ref{alg:quickscan} instead of finger search). 
Balanced binary trees can be augmented with additional pointers between all neighbouring nodes at the same tree height to support finger search.
However, maintaining these pointers introduces an additional $O(\log |CH^\ulcorner(P)|)$ overhead during updates which dominates the gains brought by using finger search.

\section{\texorpdfstring{Using the Logarithmic Method. \\}{Using the Logarithmic Method.}}
\label{sec:logarithmic}

The previous section detailed a variety of implementations of insertion-only convex hull algorithms.  
Each tree-based solution requires additional pointers and thus additional storage and overhead during operations.  
In this section, we engineer an insertion-only, vector-based solution—called \logarithmic. This structure will supports all three functionalities at the cost of increased worst-case asymptotic query time.  
Our approach builds on the logarithmic method of Overmars~\cite{overmars1983design}.  
This method transforms any static data structure with construction time $T(n)$ and decomposable query time $Q(n)$ into an insertion-only structure with amortised insertion time $O(\frac{T(n)}{n}  \log n)$. The query time for decomposable queries is $O(Q(n) \log n)$.  
Thus, we immediately obtain an amortised $O(\log^2 n)$-time insertion-only convex hull structure supporting the three decomposable queries in $O(\log^2 n)$ time, by applying any $O(n \log n)$-time static construction algorithm.

\paragraph{Our structure.}
We extend this approach to allow the use of our linear-time Graham scan and to support non-decomposable queries (see Figure~\ref{fig:logarithmic}). 
Let $\ell$ be a parameter and let $\{ B_i \mid i \in [\ell, \log n] \}$ denote a collection of arrays.
Each array $B_i$ is either empty or contains exactly $2^i$ elements. 
The exception is the smallest array $B_\ell$, which has size $2^{\ell + 1}$.
Each $B_i$ stores its points sorted by $x$-coordinate, and we also store $CH^\ulcorner(B_i)$ in an array.

To insert a new point $p$ into $P$, we place $p$ in $B_\ell$ using \competitorVector.
If $B_\ell$ is full, we trigger a \emph{merge operation}.  
A \emph{merge} identifies the largest index $j$ such that all $B_i$ for $i < j$ are non-empty.  
We perform a $j$-way merge (described below) to merge all elements from $\{ B_i \mid i \in [\ell, j] \}$ into $B_j$, filling it with $2^j$ elements sorted by $x$-coordinate.  
We then delete all arrays $B_i$ for $i < j$, and recompute $CH^\ulcorner(B_j)$ using our optimised Graham scan.
Each element participates in amortised $O(\log n)$ merges and so the amortised update time is  $O(\ell + \log n)$.

\paragraph{Implementing $j$-way merge.}
A $j$-way merge takes as input $j$ sorted arrays containing $m$ elements in total and outputs a sorted array of all $m$ elements.  
This can be done in $O(m \log j)$ time.  
Specifically, we iteratively build the result $R$ by repeatedly inserting the smallest remaining element across all input arrays.  
We maintain a pointer to the current minimum in each array and store these pointers in a min-heap.  
By popping the heap, we retrieve the next element in $O(\log j)$ time.  
Sanders~\cite{Sanders2001Fast} shows that the optimal heap implementation for this task is a \emph{loser tree}.

\paragraph{Reducing space and updates.}
We may further reduce the space usage by the following observation:

\begin{restatable}{observation}{bucket}
    Consider a point $p \in B_i$ with $p \notin CH^\ulcorner(B_i)$. Then henceforth, $p$ will never be a vertex on the quarter hull of its respective bucket.
\end{restatable}

\begin{definition}
\label{def:enclosure}
    For each $p \in B_i$ with $p \notin CH^\ulcorner(B_i)$ we say that we  \emph{witness enclosure} of $p$. 
\end{definition}

This leads to two possible bucketing schemes (we implement both) where either:
\begin{enumerate}[noitemsep]
    \item The buckets $\{ B_i \}$ partition the full input set $P$, or
    \item The buckets $\{ B_i \}$ partition only all points $p' \in P$ where did not witness the enclosure of $p'$. 
\end{enumerate}

\noindent
For the first option, we maintain a \emph{virtual size} for each $B_i$, which counts all points assigned to it.  
We explicitly store only the points in $CH^\ulcorner(B_i)$, reducing space.  
E.g., when merging $B_\ell, \ldots, B_{j-1}$ into $B_j$, we store only $CH^\ulcorner(B_j)$ but treat $B_j$ as if it held $2^j$ elements.

For the second option, we assign sequences to buckets based on convex hull size.  
E.g., when merging $B = \bigcup_{i = 0}^{j-1} B_i$, we permanently discard all points not in $CH^\ulcorner(B)$.  
If $|CH^\ulcorner(B)| \in [2^k, 2^{k+1})$, we place its points in $B_k$ in cyclic order. 
The remaining buckets $B_i$ for $i \in [0, j]$ are cleared.  
This results in fewer buckets being used, but changes when an insertion triggers a merge.

\paragraph{Tuning $B_\ell$.}
  
The parameter $\ell$ is freely tunable. Additionally, we maintained $CH^\ulcorner(P)$ using \competitorVector. We also consider using \competitorBTree instead.

\paragraph{Queries.}
We identified six query types. The first three of which are decomposable.  
Hence, any $O(\log n)$-time implementation of these queries yields an $O(\log^2 n)$-time implementation in our setting (by querying each bucket independently).
For the non-decomposable queries, no known method combines outputs from multiple $CH^\ulcorner(B_i)$ into the output on $CH^\ulcorner(P)$ in $O(\log n)$ time.  
Moreover, since there are $O(\log n)$ buckets, querying each convex hull $CH^\ulcorner(B_i)$  separately already takes $O(\log^2 n)$ total time. We show how to combine outputs for queries on $CH^\ulcorner(B_i)$ for $i \in [\lceil \log n \rceil ]$ into the output on $CH^\ulcorner(P)$ in $O(\log^2 n)$ time.  
We focus on query $4$. 
In Appendix~\ref{app:queries} we explain how to implement query 5 and  we give some comments for implementing query 6.
Denote by $A^\ulcorner(P)$ the area under the curve $CH^\ulcorner(P)$.
We show an algorithm to decide for a query point  $q$ whether $q \in A^\ulcorner(P)$. 

\begin{figure}
    \centering
    \includegraphics[width=0.9\linewidth]{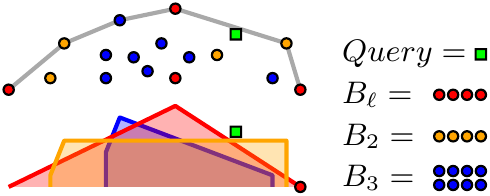}
    \caption{We partition the input $P$ across buckets (indicated by colour). We may maintain for each bucket $B_i$ the hull $CH^+(B_i)$, but we do not maintain $CH^+(P)$.}
    \label{fig:logarithmic}
\end{figure}

\newpage
Our containment query considers each $i \in [\lceil \log n \rceil ]$. We test whether $q \in A^\ulcorner(B_i)$ and consider two cases: 
\begin{enumerate}[noitemsep]
    \item If $q \in A^\ulcorner(B_i)$ then we output that $q \in A^\ulcorner(P)$.
    \item If $q \not \in A^\ulcorner(B_i)$ then $q$ is a vertex of $CH^\ulcorner(B_i \cup \{ q \})$.
\end{enumerate}

Whenever we encounter the second case, we run Algorithm~\ref{alg:quickscan} to quickly find the edges $(u_i, q)$ and $(q, v_i)$ of $CH^\ulcorner(B_i \cup \{ q \})$. 
We then observe the following: 

\begin{lemma}
\label{lem:test}
    Let there be no index $i$ for which $q \in A^\ulcorner(B_i)$.
    Let $U = \bigcup_i u_i$ and $V = \bigcup_i v_i$. Then $q \in A^\ulcorner(P)$ if and only if $q$ lies under an edge $(u, v) \in U \times V$.
\end{lemma}

\begin{proof}
    If  $q$ lies under an edge $(u, v) \in U \times V$ then it follows per the definition of an upper quarter convex hull that $q \in A^\ulcorner(P)$. 
    For proving the other direction, suppose that $q \in A^\ulcorner(P)$.
    Then $q$ lies under some edge $(x_i, y_j) \in P \times P$ and let $x_i \in B_i$ and $y_j \in B_j$.   If $x_i = u_i$ and $y_j = v_j$ then this proves the lemma.
    Otherwise, assume that $x_i \neq u_i$ and observe that $x_i$ lies left of $q$.  
   
    If $x_i \not \in U$ then $x_i \neq u_i$ and, per definition of convex hulls, $u_i$ must lie above $line(x_i, q)$.  
    Since $q$ lies under $line(x_i, y_j)$, both $x_i$ and $u_i$ lie left of $q$, and $u_i$ lies above $line(x_i, q)$, it follows that $q$ also lies under $line(u_i, y_j)$.
    If $y_j = v_j$ then this proves the lemma.
    Otherwise, we may apply the same argument to $B_j$ to conclude that $q$ must lie under $line(u_i, v_j)$ which proves the lemma.   
\end{proof}

This leads to a query algorithm:
For all  $i \in [\lceil \log n \rceil ]$, we query whether $q \in A^\ulcorner(B_i)$.
If we ever encounter case 1, the query terminates. 
Otherwise, we test for all $(u, v) \in U \times V$ whether $q$ lies under $(u, v)$ in $O(\log^2 n)$ total time. By Lemma~\ref{lem:test}, this determines whether $q \in A^\ulcorner(P)$.
In practice, we can speed this process up by computing $CH^+(U \cup V)$ instead of inspecting $U \times V$. 

\section{\texorpdfstring{Experimental Setup. \\}{Experimental Setup.}}
\label{sec:experimental_setup}

We evaluate six algorithms, all implemented in \cpp: 
\begin{itemize}[noitemsep]
    \item \competitorCHTree and \competitorCQTree from~\cite{Gaede2024Simple}, these implement the fully dynamic algorithm of~\cite{overmars1980dynamically};
    \item the three implementations from Section~\ref{sec:implementations};
    \item the \logarithmic from~Section~\ref{sec:logarithmic}.
\end{itemize}

\noindent
Whilst our theory is based on maintaining $CH^\ulcorner(P)$, our implementations maintain and query $CH^+(P)$ to be compatible with prior work. 
For these algorithms, we use the quadratic kernel that we will present in Section~\ref{sec:robustness}. We compare these implementations in terms of speed and memory usage.
All algorithms have fully robust variants based on the exact geometric predicates from CGAL (we discuss these in Appendix~\ref{app:correct:robust}).

\paragraph{Data sources.}
To our knowledge, four published works provide empirical evaluations of (static or dynamic) two-dimensional convex hull algorithms~\cite{Cadenas2019Preprocessing, gamby2018convex, Gaede2024Simple, Mei2016Gang}. We use their data sets, and more.

\paragraph{Real-world data.}
Only one type of real-world data has been previously used in convex hull experiments: The \textbf{Mammal} data sets, used in~\cite{Cadenas2019Preprocessing, Mei2016Gang}. These data sets that consist of 3D voxel representations of mammalian bodies. These are then projected to 2D~\cite{Cadenas2019Preprocessing, Mei2016Gang}.
To increase the amount of real-world data considered, we include the \textbf{Tiger} data set, used for benchmarking range queries~\cite{Govindarajan2003CRB, Danzhou2002Efficient}.
We additionally include the \textbf{Shape} set, used in clustering benchmarks~\cite{Chang2008Robust, Fu2007FLAME, Gionis2007Clustering, Jain2005Data, Zahn1971Graph}.
Although \textbf{Shape} is not real-world data, it is included here since its input size is fixed.

There is a clear reason why real-world data is rarely used in convex hull benchmarking: in practice, 2D convex hulls are typically very small. Often, they are fewer than 1,000 points which heavily skews performance. Our experiments show that, in such scenarios, the optimal insertion-only algorithm is the simple linear-time \competitorVector method.
Since we are also interested in asymptotic and worst-case behaviour, we follow prior work~\cite{Cadenas2019Preprocessing, gamby2018convex, Gaede2024Simple} and evaluate our algorithms on synthetic data, where we can control input complexity.

\paragraph{Synthetic data sets.}
To evaluate asymptotic performance, we use the synthetic data generators described in~\cite{gamby2018convex, Gaede2024Simple}, which produce 2D point sets under four different insertion patterns (we denote their expected convex hull size as a function of the input size $n$ in brackets):

\begin{itemize}[noitemsep]
    \item \textbf{Box:} Sampled uniformly from a square. \hfill ($\Theta(\log n)$)
    \item \textbf{Bell:} Sampled from a 2D Gaussian. \hfill ($\Theta(\log n)$)
    \item \textbf{Disk:} Sampled uniformly from a disk. \hfill ($\Theta(n^{1/3})$)
    \item \textbf{Circle:} Sampled uniformly from a circle. \hfill ($\Theta(n)$)
\end{itemize}
\vspace{-0.25cm}

\paragraph{Methodology.}
 Experiments were run on identical machines,  equipped with a 3.10\,GHz Intel Xeon w5-3435X processor (45\,MB cache) and 128\,GB RAM. We only compare results produced on the same machine.
Each experiment is repeated 5 times. We report the mean of these runs and schedule 4 jobs to execute in parallel and limit memory at 90 GB per process.

\subsection{\texorpdfstring{Organisation of Results. \\}{Organisation of results.}}

Section~\ref{sec:robustness} discusses the impact of using exact value computations as used in the CGAL library. 
Section~\ref{sec:parameter} presents a parameter study to identify the most efficient configuration of our \logarithmic. Specifically, we determine an appropriate value for the bottom-level bucket parameter~$B_\ell$, and evaluate alternative data structures for managing the contents of $B_\ell$. This establishes the versions of \logarithmic that we use in subsequent comparisons.
Section~\ref{sec:realworld} briefly considers real-world data sets.
Section~\ref{sec:ratio} follows~\cite{Gaede2024Simple} and fixes $n = 2^{20}$ and investigates varying query-to-insertion ratios.
Section~\ref{sec:scaling} examines the asymptotic performance of the algorithms under increasing input sizes.
Section~\ref{sec:memory} analyses memory consumption.

\newpage

\subsection{\texorpdfstring{Robustness. }{Robustness}}
\label{sec:robustness}
Algorithmic logic in geometric algorithms relies on \emph{(geometric) predicates}.
A predicate takes as input a set of objects and returns a Boolean value.
Geometric algorithms use such predicates to guide branching decisions. Update and query algorithms considered in this work rely on three predicates that take points as input:

\begin{itemize}[noitemsep]
    \item \texttt{slope}$_<((a, b), (c, d))$ tests whether the slope of ${line}(a, b)$ is strictly less than that of ${line}(c, d)$.
    \item \texttt{above\_line}$((a, b), c)$ tests whether the point $c$ lies strictly above ${line}(a, b)$.
    \item \texttt{lies\_right}$((a, b), (c, d))$ tests whether $c$ lies strictly to the right of the intersection point:
    ${line}(a, b) \cap {line}(c, d)$ (if this point exists).
\end{itemize}

\paragraph{Naïve implementation.}
Given two points $a$ and $b$, the line $line(a, b)$ can be represented as a function $f$ with slope $\frac{b.y - a.y}{b.x - a.x}$ and some intercept. Using this representation, predicates can be evaluated by computing values of $f$. For instance, \texttt{above\_line}$((a, b), c)$ can be implemented by evaluating $f$ at the $x$-coordinate of $c$ and checking whether the result is less than $c$'s $y$-coordinate.
However, this naïve approach is susceptible to precision errors: representing the slope as a floating-point value introduces rounding, and the evaluation of $f$ becomes inexact. As a result, predicates may yield incorrect results, potentially causing algorithmic failures. We refer to such failures as \emph{robustness issues}.
In convex hull maintenance, robustness issues may cause the maintained list $C$ representing $CH^+(P)$ to omit edges that belong on the hull or to include extraneous edges. For queries, they may lead to incorrect answers.
We define the \emph{Naïve} kernel as the implementation of these three predicates using this line-evaluation method. 

\paragraph{Quadratic predicates.}
To avoid such issues, one can derive formulae that evaluate the predicates directly, rather than relying on intermediate representations such as line equations. When evaluating such formulae, higher numerical precision may be required to ensure correctness. In our implementation, robustness is achieved by using at most twice the number of bits per variable (e.g., see the formulae below).
We define the \emph{Quadratic} kernel as the implementation of these three predicates using these formulae. 

\paragraph{Exact computation types.}
The Computational Geometry Algorithms Library (CGAL) provides a robust solution by offering geometric object types and native geometric predicates. 
These predicates are internally evaluated using arbitrary-precision arithmetic as needed.

\begin{align*}
 \hspace{5cm} &\texttt{slope}((a, b), (c, d)) &:=\,&  (b.y - a.y) \cdot (d.x - c.x) < (d.y - c.y) \cdot (b.x - a.x) \\
\hspace{5cm} &\texttt{above\_line}( (a, b), c) &:=\,&  (b.x - a.x)(c.y - b.y) - (c.x - b.x)(b.y - a.y) \geq 0 
\end{align*}

\newpage
  Arguably, these exact computations are CGAL's most important feature. We define the \emph{Exact} kernel as the implementation of these three predicates that invokes CGAL's exact computation.
 More robust computations can incur a computational overhead. Our experimental evaluation investigates the magnitude of this overhead and whether it is justified---i.e., whether robustness issues manifest in practice.

\subsubsection{\texorpdfstring{Experimental Evaluation. \\}{Experimental Evaluation.}}
In Appendix~\ref{app:correct:robust}, we perform an extensive analysis of the three kernels. We address two central questions:
\begin{enumerate}[nolistsep,noitemsep]
    \item Does using a weaker geometric predicate implementation introduce errors?
    \item Does using a weaker geometric predicate implementation provide a speedup?
\end{enumerate}

\paragraph{Comparing predicates using \competitorVector.}
Our first set of experiments, presented in Appendix~\ref{app:correct:robust}, compares implementations of \competitorVector using the Naïve, Quadratic, and CGAL kernels. 
We perform $2^{20}$ insertions and queries on synthetic data sets and reach the following conclusions:
First, the Naïve kernel introduces errors only rarely. The quadratic kernel incurs zero errors.  In our synthetic data sets, the Naïve kernel errors only on the \textbf{Circle} data set where it errors on all runs. 
Second, the Naïve kernel offers no speedup compared to the quadratic one. This is likely because both rely primarily on floating-point multiplication.

\paragraph{Quadratic predicates versus CGAL.}
Based on these observations, we exclude the Naïve kernel from further consideration. Next, we compare our Quadratic kernel with the exact kernel across seven considered algorithms. Due to space constraints, we defer a detailed discussion to Appendix~\ref{app:correct:robust}, and summarise the main findings here. Again the Quadratic kernel never incurs any errors.
As illustrated in Figure~\ref{fig:relative-slowdown}, all methods experience a slowdown when using the exact kernel.
If we exclude \competitorVector on \textbf{Circle} data, where shifting values in the array dominates the running time, then the slowdown is a factor $1.27$-$14.77$ with the median slowdown being a factor $4.27$. 
Since the quadratic kernel makes no errors and is faster, we  therefore conclude that an exact computation kernel is not beneficial in the context of convex hull computations. Notably, our most efficient algorithms are the most affected by changing kernel. This suggests that these methods are bottlenecked by geometric comparisons, in contrast to other algorithms that incur additional overhead from operations such as tree rotations.
Given these experimental results, we opt to use the Quadratic kernel for all algorithms throughout the remainder of this paper.

\begin{table*}
{%
\begin{tabular}[ht]{lrrrrrrrrrrrr}\toprule
&\multicolumn{3}{c}{$|B_\ell|=8$}&\multicolumn{3}{c}{$|B_\ell|=64$}&\multicolumn{3}{c}{$|B_\ell|=512$}&\multicolumn{3}{c}{$|B_\ell|=4096$}\\\cmidrule(lr){2-4}\cmidrule(lr){5-7}\cmidrule(lr){8-10}\cmidrule(lr){11-13}
&\bucketLinearGrowSuffix &\bucketBtreeSuffix&\bucketHullSuffix&\bucketLinearGrowSuffix &\bucketBtreeSuffix&\bucketHullSuffix&\bucketLinearGrowSuffix &\bucketBtreeSuffix&\bucketHullSuffix&\bucketLinearGrowSuffix &\bucketBtreeSuffix&\bucketHullSuffix\\\midrule
Insertions&\\\midrule
box&0.17&0.10&0.07&0.08&0.08&0.05&0.06&0.08&\bfseries 0.05&0.05&0.08&0.05\\
circle&0.37&0.43&0.38&\bfseries 0.34&0.39&0.42&0.66&0.40&1.10&3.43&0.41&6.61\\
\midrule
Queries\\
\midrule
box&0.14&0.11&0.09&0.10&0.10&0.08&0.08&0.10&0.08&\bfseries 0.08&0.10&0.08\\
circle&1.73&0.92&0.89&1.66&0.73&0.70&1.58&0.60&\bfseries 0.53&1.54&1.40&0.98\\\midrule
Overall\\\midrule
box&0.31&0.21&0.16&0.18&0.18&0.13&0.14&0.18&\bfseries 0.13&0.13&0.18&0.13\\
circle&2.10&1.35&1.27&2.00&1.12&1.12&2.25&\bfseries 1.00&1.62&4.97&1.80&7.59 \\
\bottomrule  
\end{tabular}
}
\caption{ 
\texorpdfstring{Our experiments consider two parameters:
the size of $|B_\ell|$ and the three implementation choices    specified at the beginning of Section~\ref{sec:parameter}. 
We show the running time in seconds and the fastest entry is \textbf{bold}.}{Our experiments consider two parameters:
the size of $|B_\ell|$ and the three implementation choices.}}
\label{tab:tuning:k}
\end{table*}

\newpage

\subsection{\texorpdfstring{Parameter study. \\}{Parameter study.}}
\label{sec:parameter}

We continue by examining the best parameters for our algorithmic implementation. 
Following~\cite{Gaede2024Simple}, we adopt a baseline input size of $n = 2^{20}$ for our parameter study. As in prior work, our query experiments focus exclusively on point-containment queries.
In Appendix~\ref{app:parameter:btree}, we determine an appropriate choice of the branching parameter $B$ for our \competitorBTree implementation. Here, we concentrate on configuring the \logarithmic using the extremal \textbf{Box} and \textbf{Circle} data sets. Specifically, we vary the size of the bottom-most bucket $B_\ell$, testing values $|B_\ell| \in \{ 8, 64, 512, 4096 \}$. Additionally, we investigate different strategies for maintaining $CH^+(B_\ell)$ within the bucket $B_\ell$:

\begin{itemize}[noitemsep]
\item \bucketLinearGrow uses \competitorVector;
    \item \bucketHull uses \competitorVector,and additionally discards all points $p' \in P$ for which enclosure has been witnessed (Def.~\ref{def:enclosure}).;
    \item \bucketBtree uses \competitorBTree with $B = 1024$.
\end{itemize}

\noindent
Table~\ref{tab:tuning:k} reports performance results for $2^{20}$ insertions, $2^{20}$ queries, and their combined total. Based on these results, we argue that choosing the size of the bucket $|B_\ell|$ to be $512$ yields the most balanced and robust performance across scenarios.
Perhaps surprisingly, \bucketHull is occasionally slower than its counterparts. We conjecture that this is because deleting points from our bucketing scheme may trigger the expensive merge operation more frequently.  

\begin{table}[H]
    \centering
    \begin{tabular}{c|c|c|c}
      Data set   &  $n$ &  $|CH(P)|$ &  $|CH(P)| / n $ \\
      \hline
        \textbf{Mammal} & 0.5-7.7M & 25-79 & $10^{-5} \,$ \% \\
   \textbf{Tiger} & 17.42-35.9M & 23-31 & $10^{-5} \,$ \% \\
      \textbf{Shape} & 240-3100 & 16-43 & 7.5-13.78 \%         
    \end{tabular}
    \caption{  Characteristics of our real-world data sets. }
    \label{tab:real_world}
\end{table}
\subsection{\texorpdfstring{Real-world data sets. \\}{Real-world data sets.}}
\label{sec:realworld}

We briefly evaluate performance on the real-world \textbf{Mammal} and \textbf{Tiger} data sets. We also include the \textbf{Shape} data sets since these are non-scalable. Each of these contains several instances with varying input sizes $n$ and convex hull sizes $|CH(P)|$ (see Table~\ref{tab:real_world}).

Observe that these hulls are \emph{very} small, which significantly skews performance results shown in Appendix~\ref{app:realworld}. 
Figure~\ref{fig:real_world} gives a representative snapshot of our results in logarithmic scale (see also Table~\ref{tab:real-world-full}).
On such small convex hull sizes, the fully dynamic algorithms \competitorCHTree and \competitorCQTree, which retain all points in $P$, are entirely non-competitive. 
The tree-based methods are fast, but are outperformed by a factor $1.50-2.00$ by our implementations of the logarithmic method. 
The median slowdown is a factor $1.79$ on these data sets.
The linear-time \competitorVector method is a factor $1-1.33$ faster than any \texttt{logarithmic} method with the median speedup being a factor $1.05$. This makes linear-time insertion marginally the best method on these data sets due to its excellent cache locality. 

\begin{figure}[h]
    \centering
    \includegraphics{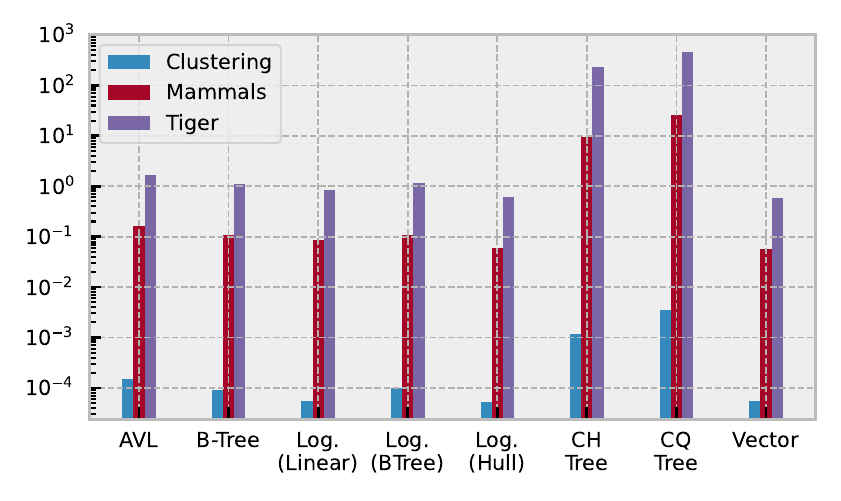}
    \caption{Average running times (s) on real data. }
    \label{fig:real_world}
\end{figure}

\subsection{\texorpdfstring{Query-to-update performance ratio.\\}{Query-to-update performance ratio.}}
\label{sec:ratio}

We evaluate the performance of our insertion-only convex hull algorithms under varying ratios of queries to updates. Following the setup of~\cite{Gaede2024Simple}, we use a baseline configuration of $2^{20}$ points and $2^{20}$ queries on synthetic data. We then vary the query-to-update ratio while keeping the total number of operations fixed at $2^{21}$.
Update inputs are drawn from four synthetic data classes, each consisting of randomly generated point sets within a bounding box of side length 1000.
Complete experimental results are presented in Appendix~\ref{app:ratio}. Representative runtime outcomes are shown in Figure~\ref{fig:compare_ratio_mix}.

\begin{figure}[H]
\includegraphics[width=\linewidth]{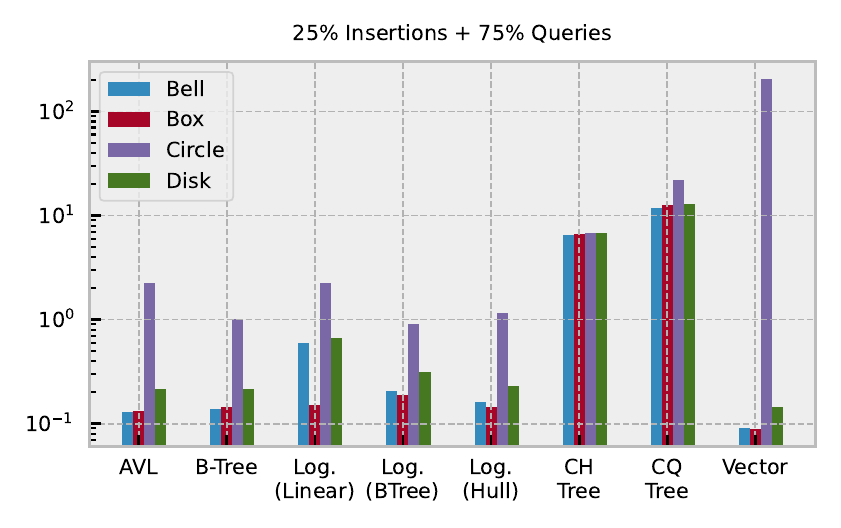}
\includegraphics[width=\linewidth]{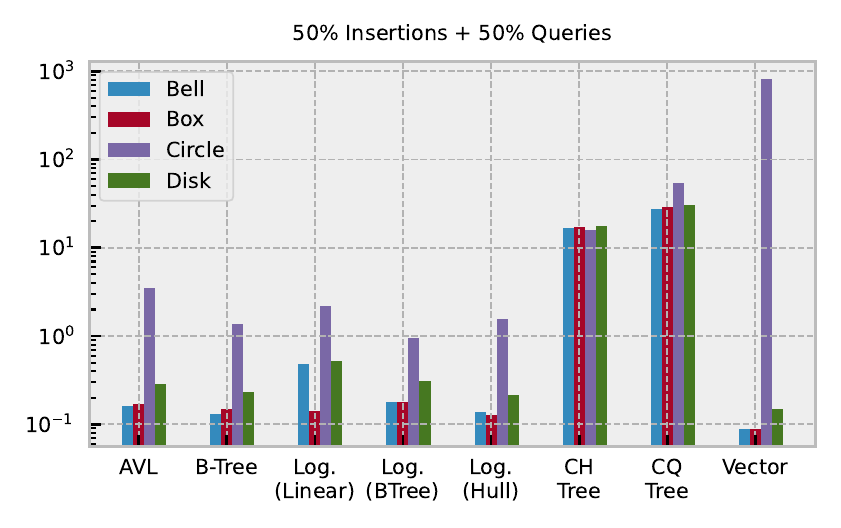}
\includegraphics[width=\linewidth]{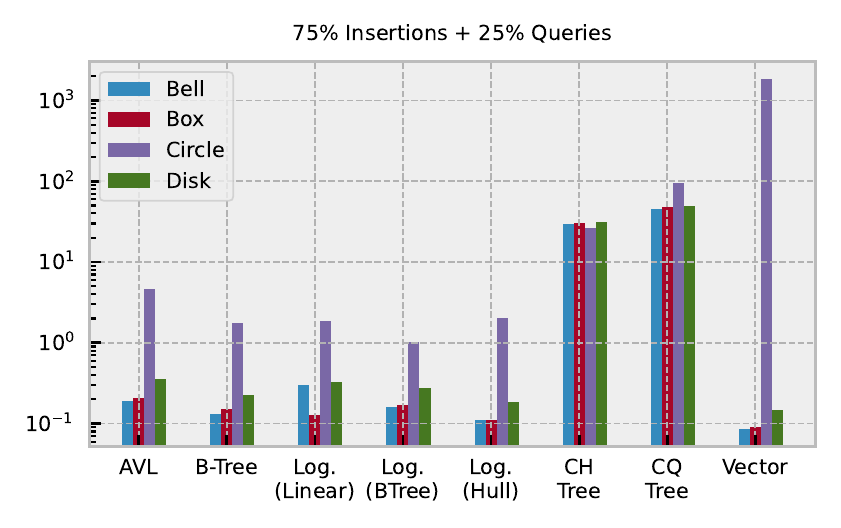}
\caption{Comparing our algorithms for different ratios of insertions to queries with $2^{21}$ total operations. }
\label{fig:compare_ratio_mix}
\end{figure}

\newpage
\paragraph{Results discussion.}
The fully dynamic \competitorCHTree and \competitorCQTree methods are orders of magnitude slower. This is expected, as these methods maintain all points in $P$ within a binary tree structure, even if they do not lie on $CH^+(P)$. They also incur significant overhead from tree balancing operations. In the remainder, we exclude them from consideration.
The comparison between the remaining algorithms is, perhaps surprisingly, largely independent of the query to insertion ratio!

The \competitorVector method, which stores only $CH^+(P)$ in a vector, is optimal in a query-only scenario. However, its insertion cost is linear in the size of $CH^+(P)$.
The extreme efficiency of its queries makes it the fastest method on the \textbf{Bell}, \textbf{Box}, and \textbf{Disk} data sets, even if there are as few as 25 percent queries. 
On \textbf{Bell}, \textbf{Box}, it outperforms all competitors by a factor $1.2$-$6.12$ with the median being a factor $1.78$. As the convex hull size increases, this gap narrows. On the \textbf{Circle} data set, where $|CH^+(P)| \approx 2^{18}$, the insertion cost dominates, and the method becomes a factor $379.72$-$1812.54$ slower than its competitors.

The tree-based \competitorAVL and \competitorBTree exhibit consistent and stable performance. Among them, \competitorBTree is slightly more efficient. However, there is no configuration where either tree-based method is the fastest. The ratio of queries to insertions has minimal effect on their performance, reflecting their $O(\log h)$ complexity for both operations. Instead, performance is strongly influenced by the size of the convex hull. When $h$ is small, i.e. on \textbf{Bell} and \textbf{Box} data, these algorithms are a factor $1.09$-$2.00$ slower than the best \texttt{Logarithmic} implementation. As $h$ grows, they fall behind and on \textbf{Circle} data these algorithms are a factor $1.67$-$4.62$ slower.

We now turn to our algorithms based on the logarithmic method. These algorithms have a theoretical amortised $O(\log n)$ insertion time and $O(\log^2 n)$ query time. Despite this asymptotic disadvantage there exists, for every tested configuration, a logarithmic-method variant that outperforms both tree-based approaches.
Among the three implementations of the logarithmic method, \bucketLinearGrow shows no clear advantage. The remaining two, \bucketHull and \bucketBtree, yield complementary strengths. \bucketHull, which maintains fewer buckets, performs better when the convex hull is small. The gap between our methods appears largely independent of the query ratio. 

We may conclude that our algorithms have distinct performance profiles that depend on the size of the convex hull, but that are largely unaffected by the ratio of queries to insertions.
The exception are the extreme cases (not depicted here) when we are nearing an insertion-only or query-only scenario.

\newpage

\subsection{\texorpdfstring{Performance with Scaling Input-size. \\}{Performance with Scaling Input-Size.}}
\label{sec:scaling}

In the previous section, we demonstrated that our algorithms exhibit different performance profiles across the four synthetic data sets. We now investigate how these performance profiles evolve as the input size increases.
For each of our seven algorithms and each of the four synthetic data sets, we define an experiment consisting of five runs. In each run, we fix the insertion-to-query ratio at 1:1 and vary the number of insertions (and thus queries) in the range $[2^{20}, 2^{26}]$. 

Appendix~\ref{app:scaling} contains the full results, with Table~\ref{tab:scaling} reporting average runtimes (in seconds) for each experiment. Due to the large number of runs, we impose a timeout of 100 seconds per run. Representative results are plotted in Figure~\ref{fig:compare_scaling_mix}, using logarithmic scale.

\paragraph{Results discussion.}
Let $n$ denote the input size and $h = |CH^+(P)|$ the size of the convex hull.
The fully dynamic \competitorCHTree and \competitorCQTree are inefficient and quickly reach the timeout as $n$ increases.
All other algorithms, with the exception of \competitorVector, exhibit logarithmic scaling behaviour. The \competitorVector method also exhibits logarithmic scaling on the \textbf{Bell} and \textbf{Box} data sets. This is because its update time being linear in $h$, and $h$ grows logarithmically with $n$ on these data sets. On the \textbf{Disk} data set, the runtime increases more steeply; however, due to its superior cache locality, the method remains among the most efficient. In contrast, on the \textbf{Circle} data set (where $h \in \Theta(n)$) \competitorVector exceeds the time limit once $n > 2^{20}$.

The tree-based \competitorAVL and \competitorBTree algorithms have update and query time complexity $O(\log h)$. This behaviour is consistent with the experimental data: these methods scale well on the \textbf{Bell} and \textbf{Box} data sets, but their performance degrades when $h$ grows more rapidly with $n$.
In particular, their scaling on the \textbf{Bell} and \textbf{Box} data sets appears to be slightly better than other algorithms. It may be that there exists an input size where these techniques are the preferred method. 

Our logarithmic method algorithms demonstrate clean logarithmic scaling in $n$ across all data sets. This observation supports our claim that the theoretical $O(\log^2 n)$ query time bound is overly pessimistic in practice.
Among these, the \bucketLinearGrow method exhibits slightly inferior scaling on data sets where the convex hull is relatively large. Although the size of the bottom bucket is fixed, linear-time insertions into this bucket appear to introduce significant overhead.

\paragraph{Overall conclusion.}
For fixed input size, our algorithms exhibit distinct performance profiles. As the input size increases, their relative performance differences remain largely consistent. With the exception of \competitorVector on the \textbf{Disk} and \textbf{Circle} data sets, all algorithms demonstrate comparable scaling behaviour.

\begin{figure}
\setlength{\lineskip}{0pt}
\vspace{-0.5cm}
\includegraphics[width=\linewidth]{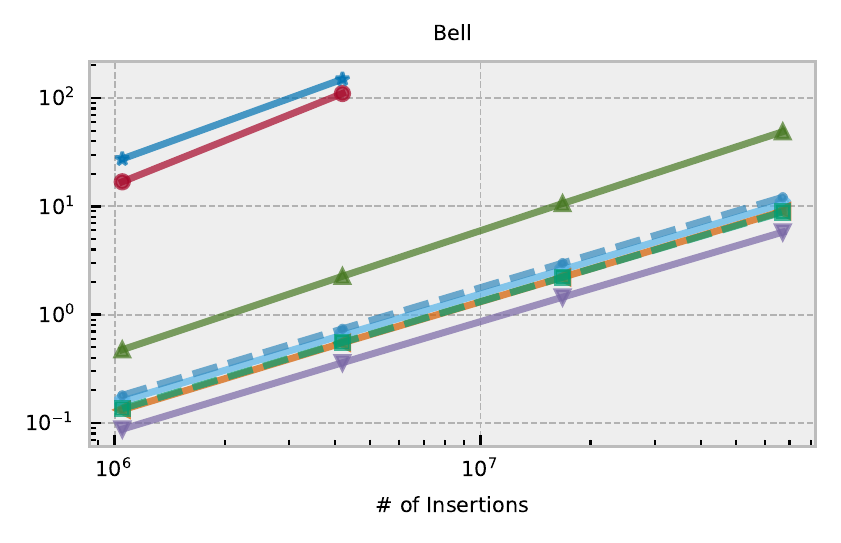}%
\vspace{-0.5cm}
\includegraphics[width=\linewidth]{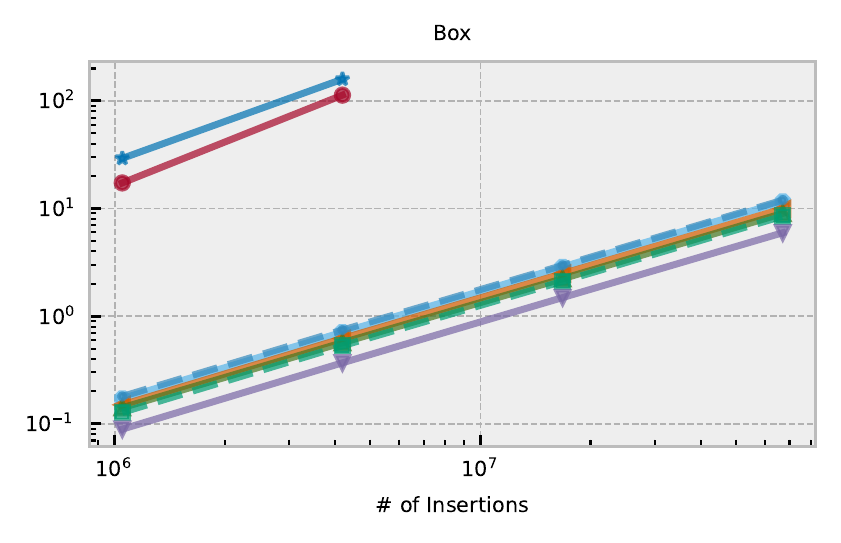}%
\vspace{-0.25cm}
\includegraphics[width=\linewidth]{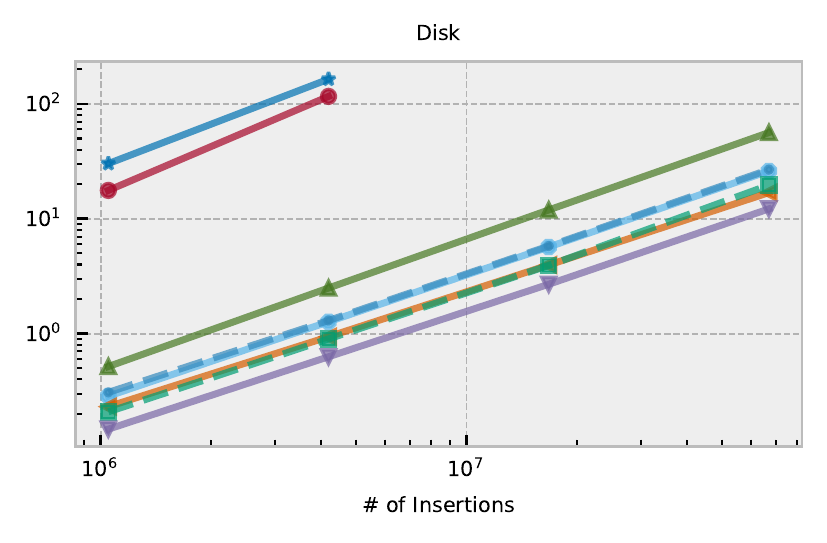}%
\vspace{-0.25cm}
\includegraphics[width=\linewidth]{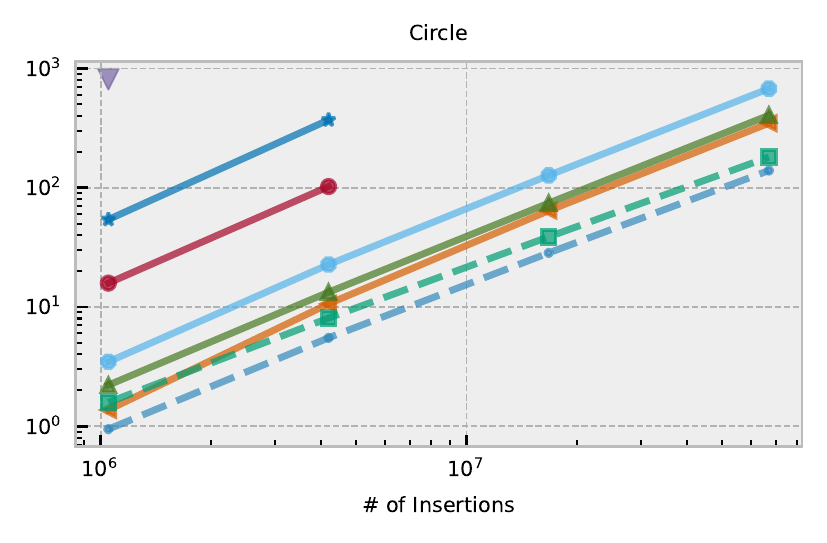}%
\vspace{-0.25cm}
\includegraphics[width=\linewidth]{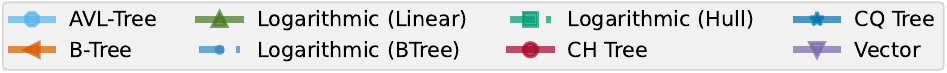}%
\vspace{0.1cm}
\caption{Comparing the various algorithms under increasing input sizes. We use a ratio of 1:1 insertions to queries.}
\label{fig:compare_scaling_mix}
\end{figure}

\newpage

\subsection{\texorpdfstring{Memory Usage. \\}{Memory Usage. }}
\label{sec:memory}

We briefly consider memory usage, focusing exclusively on the synthetic \textbf{Circle} data set, where every point lies on the convex hull, one half in $CH^+(P)$ and $CH^-(P)$ each. This scenario provides the fairest comparison for the fully dynamic \competitorCHTree and \competitorCQTree methods, which cannot discard points that do not appear on the hull.
In this experiment, we insert $2^{20}$ insertions and no queries. 
The full results, reporting peak memory consumption, are provided in Table~\ref{tab:memory_0} in Appendix~\ref{app:memory}. Averaged outcomes are plotted in  Figure~\ref{fig:compare_scaling_mix}.

\paragraph{Results.}
The \competitorVector algorithm has the smallest memory footprint and serves as our baseline. \competitorBTree uses 1.23 times more memory than the baseline, while \competitorAVL consumes 3.4 times more, due to additional pointers and balance counters. Our logarithmic methods are space-efficient, with \bucketLinearGrow being the most compact at 1.49 times the baseline. The \bucketBtreeSuffix and \bucketHullSuffix variants use 2.46 and 2.14 times the memory.
The choice of $n=2^{20}$ impacts the memory consumption of both of these methods, since it is close to a merge operation leaving  buckets empty, but allocated. For a discussion of $n=10^6$ see Appendix~\ref{app:memory}.
Finally, the fully dynamic \competitorCHTree and \competitorCQTree methods exceed the others by using 25.20 and 53.50 times more memory. %

\begin{figure}[H]
\vspace{-0.25cm}
    \centering
    \includegraphics[width=\linewidth]{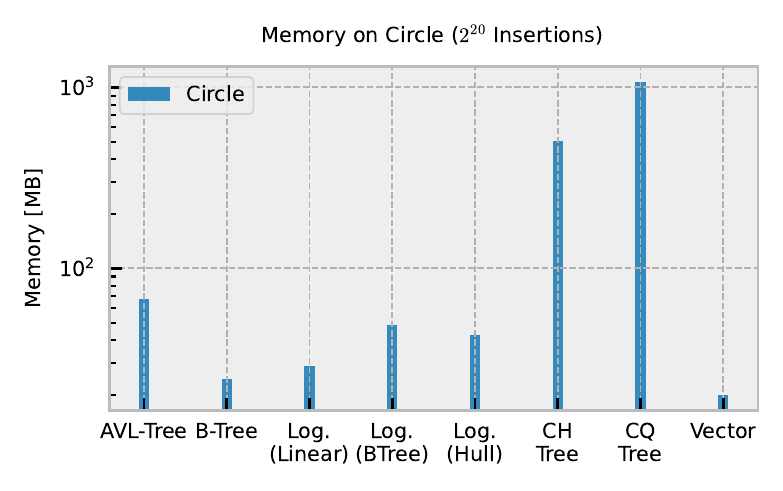}
    \caption{Memory consumption when $n  = 2^{20}$.}
    \label{fig:mem:circle}
   \vspace{-0.25cm}
\end{figure}
\section{\texorpdfstring{Conclusion. \\}{Conclusion.}}
We studied efficient methods for maintaining a convex hull in an insertion-only setting. Our work explores a broad range of algorithmic techniques and robustness considerations, supported by extensive experimental evaluation.
Motivated by the observation that contiguous memory layouts are practically advantageous, we designed a querying data structure based on the logarithmic method~\cite{overmars1983design}. 
We showed that such an implementation required substantial moderations to support the use of Graham scan, and to support non-decomposable queries (see also Appendix~\ref{app:queries}).
Our study led to several surprising insights:

In Section~\ref{sec:robustness}, we addressed robustness concerns. We showed that a theoretically non-robust quadratic kernel never produces errors in practice. The exact computation kernel provided by CGAL incurs substantial performance overhead without yielding observable accuracy benefits. For convex hull computations, we conclude that exact computation kernels offer no practical improvements and come at considerable cost.

An additional unexpected result is that on real-world data sets, the naïve linear-time \competitorVector{} method provides optimal performance. While our set of real-world data is limited, we evaluated against established benchmarks from the convex hull, range query, and clustering literature. We briefly note that other common real-world data sets, such as GPS data sets used in~\cite{buchin2020improved, UNID} also have extremely small convex hulls. These data sets will therefore provide similar~results.

On synthetic data sets, we observed several asymptotic trends. The fully dynamic methods from~\cite{Gaede2024Simple} are not competitive in the insertion-only setting. The $O(h)$ insertion time of \competitorVector{} performs surprisingly well, even on the \textbf{Disk} data set, where $h$ grows roughly as $\sqrt{n}$. Tree-based approaches offer reasonable performance due to their $O(\log h)$ update and query times, but they are never the fastest in any configuration.
Our theoretically less efficient \logarithmic consistently performs well. On nearly all instances, it closely matches the best-performing algorithm. On instances with large convex hulls, it becomes the optimal solution. We therefore regard it as the best asymptotically performing method.

Although our study focuses on the insertion-only model, it yields valuable insights into the fully dynamic setting. First, our robustness analysis carries over to fully dynamic algorithms. Second, while fully dynamic structures are slow, construction algorithms are fast. In settings where insertions and queries dominate and deletions are rare, it may be advantageous to adopt an insertion-only strategy and rebuild the hull entirely upon deletion.
Finally, our analysis highlights the poor practical scaling of tree-based structures, motivating further investigation into vector-based, fully dynamic alternatives. Our logarithmic method shows promise for extension to the fully dynamic setting: by equipping each bucket with a deletion-only structure, one could support deletions while preserving efficient queries. Insertions proceed as before—into the most recent bucket $B_\ell$—while deletions remove points from their respective buckets. Hershberger and Suri~\cite{Hershberger1992semidynamic} present an amortised deletion-only algorithm that could serve as a foundation for such an approach. 
We consider the implementation and integration of this technique with our framework a compelling direction for future work.

\newpage

\bibliographystyle{siamplain}
\bibliography{references}

\newpage
\appendix

\section{\texorpdfstring{Robustness. \\}{Robustness.}}\label{app:correct:robust}
In Section~\ref{sec:robustness}, we explained that convex hull maintenance and queries rely on geometric predicates. We introduced three kernels that implement these predicates:

\begin{enumerate}[noitemsep]
    \item A \emph{naïve kernel}, which evaluates line equations directly using floating-point arithmetic without robustness guarantees;
    \item A \emph{quadratic kernel}, which rewrites each predicate as a quadratic formula evaluated using standard floating-point arithmetic. These evaluations are robust as long as intermediate results do not overflow;
    \item An \emph{exact kernel}, which uses the CGAL library for arbitrary-precision exact predicate evaluation.
\end{enumerate}

For these three implementations, we consider the following two key questions:
(1) Does using a weaker geometric predicate implementation introduce errors?
(2) Does using a weaker geometric predicate implementation provide a speedup?

\paragraph{Comparing predicates using Graham Scan.}
We begin with a controlled experiment, implementing the \competitorVector method using each of the three kernels: Naïve, Quadratic, and CGAL. The aim of this experiment is to reduce the number of kernel configurations we consider in subsequent evaluations. We compare the three implementations on four synthetic data sets—\textbf{Box}, \textbf{Bell}, \textbf{Disk}, and \textbf{Circle}—as introduced in Section~\ref{sec:experimental_setup}. Each experiment consists of ten runs that involve $2^{20}$ insertions and queries.
Whenever we witness a deviation from the exact kernel, we terminate the run and report an error. We perform ten runs per experiment and report the fraction of runs that resulted in an error. While ten runs may seem low, each run is computationally intensive. 

Table~\ref{tab:robustness:accurate} reports the error rates observed for each kernel, and Table~\ref{tab:robustness:runtime} gives the average runtime per run in seconds.
For the \textbf{Box}, \textbf{Bell}, \textbf{Disk} data sets, where the convex hull has a significantly smaller size, none of the kernels exhibits robustness issues. On the \textbf{Circle} data set, where the convex hull contains approximately $2^{20}$ points, the Naïve kernel always has robustness issues and the Quadratic kernel is always correct. 

Since the Quadratic kernel never incurs any robustness issues, we argue that in practical scenarios, robustness issues no concern for convex hull computations. 
In terms of runtime, the Naïve and Quadratic kernels are nearly indistinguishable, most likely because both rely primarily on floating-point multiplication. In contrast, the exact kernel incurs a significant slowdown. An exception is the \textbf{Circle} data set, where insertion times are dominated by the linear-time insertion, making geometric comparisons less of a bottleneck.

\begin{table}[]
    \centering
    \begin{tabular}{|c|c|c|c|}
      Data set   & Naïve & Quadratic & Exact (CGAL) \\
      \hline
       \textbf{Box}  &  0 \% & 0 \% & 0 \% \\
       \textbf{Bell} & 0 \% & 0 \% & 0 \% \\
       \textbf{Disk}  & 0 \% & 0 \% & 0 \% \\
       \textbf{Circle} &  100 \% & 0 \% & 0 \% 
    \end{tabular}
    \caption{The error rate of our three predicate implementations on our four synthetic data sets.}
    \label{tab:robustness:accurate}
\end{table}

\begin{table}[]
    \centering
    \begin{tabular}{|c|c|c|c|}
      Data set   & Naïve & Quadratic & Exact (CGAL) \\
      \hline
       \textbf{Box}   & 0.10  & 0.9  & 0.57  \\
       \textbf{Bell} & 0.10  & 0.9  & 0.57  \\
       \textbf{Disk}  & 0.15  & 0.16  & 0.58  \\
       \textbf{Circle} &  813.85  & 813.64  & 815.17 \
    \end{tabular}
    \caption{The runtime in seconds of \competitorVector after $2^{20}$ insertions and queries using our three predicate implementations on our four synthetic data sets.}
    \label{tab:robustness:runtime}
\end{table}

\begin{figure}[H]
\vspace{-0.25cm}
    \centering
    \includegraphics[width=\linewidth]{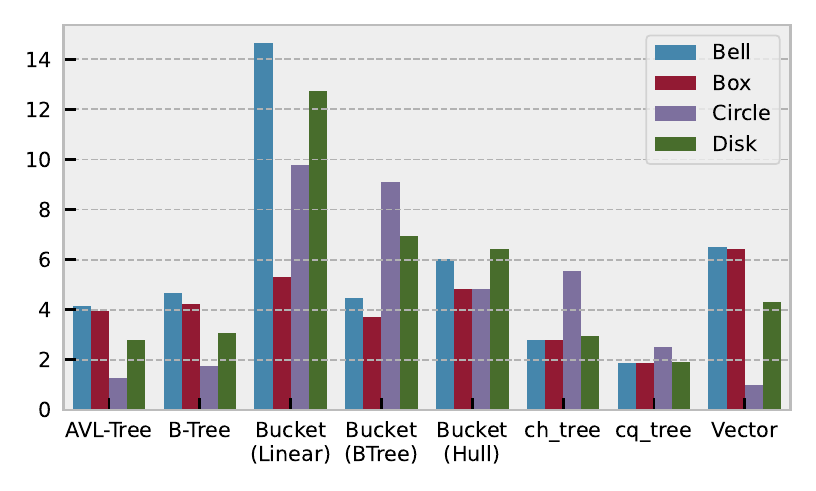}
    \caption{Relative Slowdown for each implementation (Lower is better). $2^{20}$ insertions and queries each.}
    \label{fig:relative-slowdown}
\vspace{-0.25cm}
\end{figure}
\vspace{-0.25cm}
\paragraph{Quadratic versus Exact kernels.}
Since the Naïve kernel offers no benefit over the Quadratic kernel, we discard it from further evaluation. We examine the overhead of using CGAL’s exact kernel versus the Quadratic kernel across all seven of our algorithms. The full results are presented in Table~\ref{tab:res:robustness}. Figure~\ref{fig:relative-slowdown} shows the relative slowdown on a logarithmic scale.
Several insights emerge from this comparison:

First, the slowdown introduced by the exact kernel depends strongly on both the algorithm and the data set. Algorithms that store points contiguously in memory (e.g., \bucketNameBase methods and \competitorVector) suffer far greater slowdowns than tree-based methods. %
Second, all algorithms exhibit a noticeable performance degradation when using the exact kernel.
Our \bucketNameBase algorithms pay a higher price since more exact queries are needed in every separate bucket.
Given that the Quadratic kernel incurs no correctness issues on all experiments, we conclude that the CGAL exact kernel is overly expensive in the context of convex hull maintenance and query workloads. As such, we adopt the Quadratic kernel for all subsequent experiments.

\begin{sidewaystable}[]
    \centering
    \begin{tabular}{lrrrrrrrrrrrrrrrrrr}
    \toprule%
    &\multicolumn{4}{c}{}&\multicolumn{6}{c}{\bucketNameBase}\\\cmidrule(lr){6-11}
    &\multicolumn{2}{c}{\competitorAVL}&\multicolumn{2}{c}{\competitorBTree}&\multicolumn{2}{c}{\bucketLinearGrowSuffix}&\multicolumn{2}{c}{\bucketBtreeSuffix}&\multicolumn{2}{c}{\bucketHullSuffix}&\multicolumn{2}{c}{\competitorCHTree}&\multicolumn{2}{c}{\competitorCQTree}&\competitorTerrible&\multicolumn{2}{c}{\competitorVector}\\\cmidrule(lr){2-3}\cmidrule(lr){4-5}\cmidrule(lr){6-7}\cmidrule(lr){8-9}\cmidrule(lr){10-11}\cmidrule(lr){12-13}\cmidrule(lr){14-15}\cmidrule(lr){16-16}\cmidrule(lr){17-18}
    &E&I&E&I&E&I&E&I&E&I&E&I&E&I&&E&I\\\midrule
bell&0.63&0.15&0.62&0.13&6.94&0.47&0.81&0.18&0.82&0.14&45.46&16.35&51.22&27.07&0.10&0.57&\bfseries 0.09\\
box&0.65&0.17&0.64&0.15&0.75&0.14&0.67&0.18&0.61&0.13&47.43&16.86&54.26&28.66&0.10&0.58&\bfseries 0.09\\
circle&4.46&3.50&2.35&1.35&21.68&2.22&8.52&\bfseries 0.94&7.62&1.58&88.67&15.98&134.85&53.93&813.83&815.31&813.87\\
disk&0.76&0.27&0.71&0.23&6.64&0.52&2.13&0.31&1.37&0.21&51.05&17.20&57.68&29.90&0.16&0.64&\bfseries 0.15\\

\bottomrule
    \end{tabular}
    \caption{I denotes the inexact method using the quadratic predicate kernel. E denotes the exact computation kernel using CGAL. }
    \label{tab:res:robustness}
\end{sidewaystable}

\newpage

\section{\texorpdfstring{Data for real-world data sets. \\}{Data for real-world data sets.}}
\label{app:realworld}

We briefly evaluate performance of our insertion-only algorithms on the real-world \textbf{Mammal} and \textbf{Tiger} data sets. We also include the \textbf{Shape} data set, which, while synthetic, is non-scalable and hence included under the umbrella term “real-world” for convenience. Each data set comprises multiple instances with varying input sizes $n$ and convex hull sizes $|CH(P)|$.

Table~\ref{tab:real_world_2} summarises the key characteristics of these data sets. The primary observation is that the convex hull is \emph{very} small across all instances, which strongly influences the relative performance of our algorithms.

\begin{table}[H]
    \centering
    \begin{tabular}{c|c|c|c}
      Data set   &  $n$ &  $|CH(P)|$ &  $|CH(P)| / n $ \\
      \hline
        \textbf{Mammal} & 0.5-7.7M & 25-79 & $10^{-5} \,$ \% \\
   \textbf{Tiger} & 17.42-35.9M & 23-31 & $10^{-5} \,$ \% \\
      \textbf{Shape} & 240-3100 & 16-43 & 7.5-13.78 \%         
    \end{tabular}
    \caption{  Characteristics of our real-world data sets. }
    \label{tab:real_world_2}
\end{table}

\begin{table*}\vspace{-0.75cm}
{%
    \begin{tabular}{llS[table-format=2.2]S[table-format=2.2]S[table-format=2.2]S[table-format=2.2]S[table-format=2.2]S[table-format=2.2]S[table-format=2.2]S[table-format=4.2]}
    \toprule 
    &&&&\multicolumn{3}{c}{{{{\bucketNameBase}}}}\\\cmidrule(lr){5-7}
    Type&Name&{{{\competitorAVL}}}&{{{\competitorBTree}}}&{{{\bucketLinearGrowSuffix}}}&{{{\bucketBtreeSuffix}}}&{{{\bucketHullSuffix}}}&{{{\competitorCHTree}}}&{{{\competitorCQTree}}}&{{{\competitorVector}}}\\
    \midrule
Tiger&tiger2006east&1.65&1.12&0.84&1.15&0.62&229.51&458.01&\bfseries 0.59\\
Tiger&tiger2006se&3.30&2.20&1.74&2.26&1.27&{{{DNF}}}&972.50&\bfseries 1.20\\
Clustering&Aggregation.txt&0.00&0.00&0.00&0.00&0.00&0.00&0.01&\bfseries 0.00\\
Clustering&Compound.txt&0.00&0.00&0.00&0.00&\bfseries 0.00&0.00&0.00&0.00\\
Clustering&D31.txt&0.00&0.00&0.00&0.00&0.00&0.01&0.02&\bfseries 0.00\\
Clustering&R15.txt&0.00&0.00&\bfseries 0.00&0.00&0.00&0.00&0.00&0.00\\
Clustering&flame.txt&0.00&0.00&\bfseries 0.00&0.00&0.00&0.00&0.00&0.00\\
Clustering&jain.txt&0.00&0.00&0.00&0.00&0.00&0.00&0.00&\bfseries 0.00\\
Clustering&pathbased.txt&0.00&0.00&0.00&0.00&\bfseries 0.00&0.00&0.00&0.00\\
Clustering&spiral.txt&0.00&0.00&\bfseries 0.00&0.00&0.00&0.00&0.00&0.00\\
Mammal&bison:xy&0.21&0.14&0.12&0.15&0.08&11.45&32.92&\bfseries 0.08\\
Mammal&bison:xz&0.20&0.14&0.12&0.15&0.08&11.48&33.61&\bfseries 0.07\\
Mammal&bison:yz&0.20&0.14&0.11&0.15&0.08&{{{OOM}}}&43.30&\bfseries 0.07\\
Mammal&bull:xy&0.14&0.09&0.08&0.09&0.05&6.98&21.06&\bfseries 0.05\\
Mammal&bull:xz&0.14&0.10&0.08&0.10&0.05&7.17&21.26&\bfseries 0.05\\
Mammal&bull:yz&0.14&0.09&0.07&0.09&0.05&11.71&27.29&\bfseries 0.05\\
Mammal&camel:xy&0.21&0.14&0.11&0.15&0.08&11.30&32.66&\bfseries 0.07\\
Mammal&camel:xz&0.22&0.15&0.11&0.15&0.08&11.32&32.50&\bfseries 0.07\\
Mammal&camel:yz&0.20&0.14&0.11&0.14&0.08&19.48&43.48&\bfseries 0.07\\
Mammal&el'\_african:xy&0.79&0.53&0.38&0.55&0.29&49.28&130.65&\bfseries 0.27\\
Mammal&el'\_african:xz&0.82&0.54&0.38&0.55&0.29&49.30&128.60&\bfseries 0.27\\
Mammal&el'\_african:yz&0.78&0.54&0.41&0.55&0.29&105.93&200.36&\bfseries 0.28\\
Mammal&el'\_asian:xy&0.52&0.35&0.25&0.36&0.19&30.58&82.40&\bfseries 0.18\\
Mammal&el'\_asian:xz&0.51&0.34&0.25&0.35&0.19&30.39&81.34&\bfseries 0.18\\
Mammal&el'\_asian:yz&0.52&0.33&0.26&0.34&0.19&59.27&119.60&\bfseries 0.18\\
Mammal&elk:xy&0.16&0.11&0.09&0.11&0.06&8.70&25.13&\bfseries 0.06\\
Mammal&elk:xz&0.17&0.11&0.09&0.11&0.06&8.88&25.19&\bfseries 0.06\\
Mammal&elk:yz&0.16&0.10&0.09&0.11&0.06&13.54&31.45&\bfseries 0.06\\
Mammal&giraffe:xy&0.22&0.16&0.12&0.16&0.08&13.56&37.23&\bfseries 0.08\\
Mammal&giraffe:xz&0.23&0.15&0.12&0.15&0.08&13.55&37.00&\bfseries 0.08\\
Mammal&giraffe:yz&0.22&0.15&0.12&0.15&0.08&20.76&46.73&\bfseries 0.08\\
Mammal&horse:xy&0.18&0.12&0.10&0.12&0.07&9.88&27.89&\bfseries 0.06\\
Mammal&horse:xz&0.18&0.12&0.10&0.12&0.06&10.07&27.96&\bfseries 0.06\\
Mammal&horse:yz&0.19&0.12&0.09&0.12&0.07&15.83&35.55&\bfseries 0.06\\
Mammal&pig:xy&0.05&0.04&0.03&0.04&0.02&2.38&7.46&\bfseries 0.02\\
Mammal&pig:xz&0.05&0.04&0.03&0.04&0.02&2.46&7.54&\bfseries 0.02\\
Mammal&pig:yz&0.05&0.03&0.03&0.03&0.02&3.37&8.78&\bfseries 0.02\\
Mammal&polar:xy&0.09&0.06&0.04&0.06&0.03&3.88&11.73&\bfseries 0.03\\
Mammal&polar:xz&0.08&0.05&0.04&0.06&0.03&4.07&11.90&\bfseries 0.03\\
Mammal&polar:yz&0.09&0.06&0.04&0.06&0.03&5.40&13.80&\bfseries 0.03\\
Mammal&redeer:xy&0.08&0.05&0.03&0.05&0.03&3.29&9.39&\bfseries 0.03\\
Mammal&redeer:xz&0.08&0.05&0.03&0.05&0.03&3.43&9.51&\bfseries 0.03\\
Mammal&redeer:yz&0.08&0.05&0.04&0.05&0.03&4.11&9.49&\bfseries 0.03\\
Mammal&reindeer:xy&0.09&0.05&0.04&0.06&0.03&4.11&12.21&\bfseries 0.03\\
Mammal&reindeer:xz&0.08&0.05&0.05&0.06&0.03&4.09&12.17&\bfseries 0.03\\
Mammal&reindeer:yz&0.09&0.06&0.05&0.06&0.03&5.52&14.46&\bfseries 0.03\\
Mammal&rhino:xy&0.15&0.11&0.09&0.11&0.06&8.66&24.81&\bfseries 0.06\\
Mammal&rhino:xz&0.16&0.10&0.09&0.11&0.06&8.52&24.41&\bfseries 0.06\\
Mammal&rhino:yz&0.16&0.11&0.09&0.11&0.06&13.19&31.20&\bfseries 0.06\\
Mammal&tapir:xy&0.10&0.07&0.05&0.07&0.04&4.65&14.18&\bfseries 0.03\\
Mammal&tapir:xz&0.10&0.07&0.05&0.07&0.04&4.77&14.14&\bfseries 0.03\\
Mammal&tapir:yz&0.09&0.06&0.05&0.06&0.04&6.83&17.05&\bfseries 0.03\\

  \bottomrule
  \end{tabular}
    }
    \caption{Average running time per insert or query on real world instance in seconds. Bold entries are the optimal running time. DNF denotes that the entire experiment did not finish within 24 hours. OOM if more than 90 GB of memory was used. }
    \label{tab:real-world-full}
    
\end{table*}

\paragraph{Result discussion.}
Table~\ref{tab:real-world-full} reports the complete set of experimental results. Each experiment evaluates a single algorithm across the various data sets, with results averaged over five runs per instance. The reported values represent average runtime per instance in seconds. Below, we discuss the performance of the algorithms in increasing order of efficiency.

The fully dynamic \competitorCHTree and \competitorCQTree algorithms, which retain all points in $P$, are entirely non-competitive on these data sets. These methods maintain binary trees with millions of points, despite the convex hull containing fewer than 100 points in most cases.
We must emphasise that these methods are fully dynamic in contrast to the subsequent faster methods. 

The tree-based \competitorAVL and \competitorBTree algorithms, which store only the convex hull points, perform substantially better.

Surprisingly, our algorithms based on the logarithmic method outperform the tree-based approaches. Recall that \bucketLinearGrow and \bucketBtree partition the full point set $P$ into $O(\log n)$ buckets.
Denote by $n$ the size of $P$ and by $h$ the size of $CH^+(P)$. 
Although the amortised update time for these methods is $O(\log n)$ compared to $O(\log h)$ for the tree-based algorithms, and $h << n$, the overhead from tree balancing appears to outweigh the theoretical advantage. This underscores the substantial runtime cost associated with maintaining balanced trees.

Among the logarithmic methods, \bucketHull achieves the best performance. This is to be expected, as \bucketHull discards enclosed points and thus maintains a more compact structure. In this setting where the convex hulls are very small compared to $n$, this significantly reducing both space usage and computational overhead. Nevertheless, all logarithmic-method variants remain competitive with one another.

Finally, the linear-time \competitorVector method emerges as the clear winner on these data sets, owing to its excellent cache locality and low overhead. Given that the convex hull is typically small in real-world data, we argue that \competitorVector is an optimal choice in such scenarios.

\section{\texorpdfstring{Varying the query-to-update ratio. \\ }{Varying the query-to-update ratio.}}
\label{app:ratio}

We evaluate the performance of our insertion-only convex hull algorithms under varying ratios of queries to updates. Following the setup of~\cite{Gaede2024Simple}, we use a baseline configuration of $2^{20}$ points and $2^{20}$ queries. We then vary the query-to-update ratio while keeping the total number of operations fixed at $2^{21}$.
Update inputs are drawn from four synthetic data classes, each consisting of randomly generated point sets within a bounding box of side length 1000.

\paragraph{Results discussion.}
Across all data sets, the fully dynamic \competitorCHTree and \competitorCQTree methods are orders of magnitude slower. This is expected, as these methods maintain all points in $P$ within a binary tree structure, even if they do not lie on $CH^+(P)$. They also incur significant overhead from frequent tree balancing operations.
The comparison between the remaining algorithms is, perhaps surprisingly, largely independent of the query to insertion ratio!

The \competitorVector method, which stores only $CH^+(P)$ in a vector, is optimal in a query-only scenario. However, its insertion cost is linear in the size of $CH^+(P)$. The extreme efficiency of its queries makes it the fastest method on the \textbf{Bell}, \textbf{Box}, and \textbf{Disk} data sets, even if there are as few as 25 percent queries. 
On \textbf{Bell}, \textbf{Box}, it outperforms all competitors by a factor $1.2$-$6.12$ with the median being a factor $1.78$. As the convex hull size increases, this gap narrows. On the \textbf{Circle} data set, where $|CH^+(P)| \approx 2^{18}$, the insertion cost dominates, and the method becomes a factor $379.72$-$1812.54$ slower than its competitors.

The tree-based \competitorAVL and \competitorBTree algorithms exhibit consistent and stable performance. Among them, \competitorBTree is slightly more efficient. However, there is no configuration where either tree-based method is the fastest. The ratio of queries to insertions has minimal effect on their performance, reflecting their $O(\log h)$ complexity for both operations. Instead, performance is strongly influenced by the size of the convex hull. When $h$ is small, i.e. on \textbf{Bell} and \textbf{Box} data, these algorithms are a factor $1.09$-$2.00$ slower than the best \texttt{Logarithmic} implementation. As $h$ grows, they fall behind and on \textbf{Circle} data these algorithms are a factor $1.67$-$4.62$ slower.

We now turn to our algorithms based on the logarithmic method. These algorithms have a theoretical amortised $O(\log n)$ insertion time and $O(\log^2 n)$ query time. Despite this asymptotic disadvantage there exists, for every tested configuration, a logarithmic-method variant that outperforms both tree-based approaches.
Among the three implementations of the logarithmic method, \bucketLinearGrow shows no clear advantage. The remaining two, \bucketHull and \bucketBtree, yield complementary strengths. \bucketHull, which deletes any point $p' \in P$ for which it witnesses enclosure, performs better when the convex hull is small. The gap between our methods appears largely independent of the query ratio. As the convex hull becomes large, \bucketHull becomes increasingly expensive.

We may conclude that our algorithms have distinct performance profiles that depend on the size of the convex hull, but that are largely unaffected by the ratio of queries to insertions.
The exception are the extreme cases (not depicted here) when we are nearing an insertion-only or query-only scenario. 
\newpage

\begin{table}[H]
{%
    \begin{tabular}{lllS[table-format=2.2]S[table-format=2.2]S[table-format=2.2]S[table-format=2.2]S[table-format=2.2]S[table-format=2.2]S[table-format=2.2]S[table-format=4.2]}
    \toprule 
    &&\# True &&&\multicolumn{3}{c}{{{{\bucketNameBase}}}}\\\cmidrule(lr){6-8}
    &Seed&Queries&{{{\competitorAVL}}}&{{{\competitorBTree}}}&{{{\bucketLinearGrowSuffix}}}&{{{\bucketBtreeSuffix}}}&{{{\bucketHullSuffix}}}&{{{\competitorCHTree}}}&{{{\competitorCQTree}}}&{{{\competitorVector}}}\\
    \midrule
bell&\numprint{0}&\numprint{398307}&0.19&0.12&0.29&0.15&0.11&29.55&44.54&\bfseries 0.08\\
bell&\numprint{1}&\numprint{405326}&0.19&0.13&0.29&0.15&0.11&29.54&44.82&\bfseries 0.09\\
bell&\numprint{2}&\numprint{398217}&0.19&0.14&0.30&0.17&0.11&29.54&44.90&\bfseries 0.09\\
bell&\numprint{3}&\numprint{392900}&0.20&0.14&0.30&0.17&0.11&29.50&45.30&\bfseries 0.09\\
bell&\numprint{4}&\numprint{372110}&0.18&0.12&0.31&0.15&0.11&29.46&45.23&\bfseries 0.08\\
bell&\numprint{5}&\numprint{383752}&0.19&0.14&0.29&0.16&0.11&29.53&45.35&\bfseries 0.09\\
bell&\numprint{6}&\numprint{387694}&0.19&0.13&0.31&0.16&0.11&29.45&45.19&\bfseries 0.09\\
bell&\numprint{7}&\numprint{411733}&0.20&0.14&0.30&0.17&0.12&29.51&45.15&\bfseries 0.09\\
bell&\numprint{8}&\numprint{414112}&0.18&0.14&0.29&0.16&0.11&29.44&45.21&\bfseries 0.08\\
bell&\numprint{9}&\numprint{378721}&0.19&0.13&0.31&0.16&0.11&29.36&44.96&\bfseries 0.09\\
box&\numprint{0}&\numprint{524270}&0.22&0.16&0.13&0.18&0.12&30.46&47.54&\bfseries 0.10\\
box&\numprint{1}&\numprint{524281}&0.20&0.14&0.13&0.16&0.10&30.37&47.94&\bfseries 0.08\\
box&\numprint{2}&\numprint{524272}&0.20&0.14&0.13&0.16&0.11&30.22&47.95&\bfseries 0.09\\
box&\numprint{3}&\numprint{524275}&0.20&0.15&0.12&0.16&0.11&30.20&47.99&\bfseries 0.09\\
box&\numprint{4}&\numprint{524270}&0.22&0.17&0.13&0.19&0.12&30.24&47.81&\bfseries 0.09\\
box&\numprint{5}&\numprint{524274}&0.21&0.16&0.13&0.18&0.11&30.31&47.98&\bfseries 0.09\\
box&\numprint{6}&\numprint{524278}&0.21&0.15&0.13&0.17&0.11&30.42&48.13&\bfseries 0.09\\
box&\numprint{7}&\numprint{524267}&0.21&0.15&0.13&0.17&0.11&30.23&48.08&\bfseries 0.09\\
box&\numprint{8}&\numprint{524276}&0.21&0.15&0.12&0.16&0.11&30.28&48.04&\bfseries 0.08\\
box&\numprint{9}&\numprint{524279}&0.22&0.15&0.13&0.17&0.12&30.30&48.20&\bfseries 0.10\\
circle&\numprint{0}&\numprint{411435}&4.51&1.70&1.86&\bfseries 1.02&2.03&26.55&96.09&1831.32\\
circle&\numprint{1}&\numprint{411825}&4.82&1.79&1.86&\bfseries 1.02&2.03&26.92&92.88&1830.27\\
circle&\numprint{2}&\numprint{411096}&4.64&1.80&1.86&\bfseries 1.04&2.04&26.47&94.55&1830.16\\
circle&\numprint{3}&\numprint{412442}&4.65&1.68&1.85&\bfseries 1.01&2.02&27.13&94.86&1829.72\\
circle&\numprint{4}&\numprint{411446}&4.54&1.76&1.86&\bfseries 1.04&2.05&26.42&93.98&1832.46\\
circle&\numprint{5}&\numprint{411565}&4.66&1.75&1.84&\bfseries 1.01&2.02&26.24&94.90&1830.61\\
circle&\numprint{6}&\numprint{411744}&4.55&1.79&1.85&\bfseries 1.01&2.02&27.17&96.92&1830.67\\
circle&\numprint{7}&\numprint{411782}&4.68&1.75&1.85&\bfseries 1.03&2.04&26.40&95.58&1830.21\\
circle&\numprint{8}&\numprint{412035}&4.76&1.65&1.84&\bfseries 1.03&2.04&26.44&95.93&1832.21\\
circle&\numprint{9}&\numprint{411829}&4.68&1.76&1.87&\bfseries 1.03&2.04&26.80&95.38&1830.05\\
disk&\numprint{0}&\numprint{412153}&0.36&0.23&0.32&0.28&0.18&31.21&49.41&\bfseries 0.15\\
disk&\numprint{1}&\numprint{411312}&0.36&0.23&0.33&0.27&0.19&31.05&49.79&\bfseries 0.15\\
disk&\numprint{2}&\numprint{411847}&0.36&0.23&0.32&0.27&0.18&30.86&49.61&\bfseries 0.15\\
disk&\numprint{3}&\numprint{411547}&0.36&0.23&0.32&0.27&0.18&30.89&49.83&\bfseries 0.15\\
disk&\numprint{4}&\numprint{412199}&0.36&0.22&0.32&0.27&0.18&30.98&49.68&\bfseries 0.15\\
disk&\numprint{5}&\numprint{411501}&0.36&0.23&0.32&0.28&0.18&31.07&49.82&\bfseries 0.15\\
disk&\numprint{6}&\numprint{411963}&0.35&0.23&0.33&0.27&0.18&30.96&49.97&\bfseries 0.15\\
disk&\numprint{7}&\numprint{412072}&0.36&0.23&0.32&0.27&0.18&30.91&49.56&\bfseries 0.15\\
disk&\numprint{8}&\numprint{411629}&0.36&0.23&0.33&0.28&0.18&31.09&49.93&\bfseries 0.15\\
disk&\numprint{9}&\numprint{411977}&0.36&0.22&0.32&0.27&0.18&30.95&49.85&\bfseries 0.15\\

  \bottomrule
  \end{tabular}
    }
    \caption{Ratio: 25\% insertions 75\% queries.}
    \label{tab:mixing:25:75:all}
\end{table}

\newpage

-

\newpage

\begin{table}[H]
{%
    \begin{tabular}{lllS[table-format=2.2]S[table-format=2.2]S[table-format=2.2]S[table-format=2.2]S[table-format=2.2]S[table-format=2.2]S[table-format=2.2]S[table-format=4.2]}
    \toprule 
    &&\# True &&&\multicolumn{3}{c}{{{{\bucketNameBase}}}}\\\cmidrule(lr){6-8}
    &Seed&Queries&{{{\competitorAVL}}}&{{{\competitorBTree}}}&{{{\bucketLinearGrowSuffix}}}&{{{\bucketBtreeSuffix}}}&{{{\bucketHullSuffix}}}&{{{\competitorCHTree}}}&{{{\competitorCQTree}}}&{{{\competitorVector}}}\\
    \midrule
bell&\numprint{0}&\numprint{835488}&0.16&0.12&0.49&0.17&0.13&17.08&27.30&\bfseries 0.08\\
bell&\numprint{1}&\numprint{799267}&0.16&0.13&0.46&0.17&0.14&16.98&27.40&\bfseries 0.09\\
bell&\numprint{2}&\numprint{780684}&0.16&0.14&0.48&0.19&0.14&16.95&27.47&\bfseries 0.09\\
bell&\numprint{3}&\numprint{805031}&0.17&0.13&0.47&0.18&0.14&16.92&27.66&\bfseries 0.09\\
bell&\numprint{4}&\numprint{731322}&0.16&0.13&0.52&0.17&0.14&16.95&27.55&\bfseries 0.09\\
bell&\numprint{5}&\numprint{824217}&0.16&0.14&0.47&0.19&0.14&16.99&27.74&\bfseries 0.09\\
bell&\numprint{6}&\numprint{789435}&0.16&0.13&0.46&0.18&0.14&16.97&27.53&\bfseries 0.09\\
bell&\numprint{7}&\numprint{813433}&0.16&0.14&0.46&0.19&0.14&16.86&27.49&\bfseries 0.09\\
bell&\numprint{8}&\numprint{812291}&0.16&0.13&0.45&0.17&0.13&16.99&27.63&\bfseries 0.09\\
bell&\numprint{9}&\numprint{773955}&0.16&0.14&0.50&0.19&0.14&16.92&27.27&\bfseries 0.09\\
box&\numprint{0}&\numprint{1048537}&0.18&0.17&0.14&0.20&0.14&17.57&29.19&\bfseries 0.10\\
box&\numprint{1}&\numprint{1048552}&0.16&0.14&0.14&0.17&0.12&17.50&29.16&\bfseries 0.08\\
box&\numprint{2}&\numprint{1048537}&0.16&0.14&0.14&0.17&0.12&17.36&29.23&\bfseries 0.08\\
box&\numprint{3}&\numprint{1048548}&0.17&0.15&0.15&0.18&0.12&17.35&29.30&\bfseries 0.08\\
box&\numprint{4}&\numprint{1048539}&0.18&0.15&0.14&0.18&0.13&17.44&29.13&\bfseries 0.09\\
box&\numprint{5}&\numprint{1048537}&0.18&0.16&0.14&0.19&0.13&17.44&29.27&\bfseries 0.09\\
box&\numprint{6}&\numprint{1048549}&0.17&0.15&0.14&0.18&0.13&17.45&29.29&\bfseries 0.09\\
box&\numprint{7}&\numprint{1048528}&0.18&0.15&0.14&0.17&0.13&17.38&29.20&\bfseries 0.09\\
box&\numprint{8}&\numprint{1048556}&0.17&0.14&0.14&0.17&0.12&17.46&29.35&\bfseries 0.08\\
box&\numprint{9}&\numprint{1048552}&0.18&0.16&0.13&0.19&0.14&17.35&29.23&\bfseries 0.10\\
circle&\numprint{0}&\numprint{823191}&3.51&1.36&2.21&\bfseries 0.95&1.59&15.86&54.55&813.96\\
circle&\numprint{1}&\numprint{823602}&3.69&1.35&2.21&\bfseries 0.84&1.50&16.17&53.94&813.97\\
circle&\numprint{2}&\numprint{822796}&3.37&1.39&2.22&\bfseries 0.96&1.59&15.77&54.33&813.71\\
circle&\numprint{3}&\numprint{824056}&3.46&1.34&2.20&\bfseries 0.93&1.58&16.18&54.36&813.26\\
circle&\numprint{4}&\numprint{823655}&3.45&1.37&2.22&\bfseries 0.97&1.60&15.96&53.27&814.68\\
circle&\numprint{5}&\numprint{823865}&3.57&1.37&2.21&\bfseries 0.95&1.58&15.55&54.85&813.54\\
circle&\numprint{6}&\numprint{823992}&3.49&1.36&2.21&\bfseries 0.95&1.59&16.23&55.71&813.83\\
circle&\numprint{7}&\numprint{823296}&3.44&1.38&2.21&\bfseries 0.96&1.60&15.82&54.78&813.95\\
circle&\numprint{8}&\numprint{824044}&3.68&1.31&2.20&\bfseries 0.96&1.59&16.00&54.87&814.43\\
circle&\numprint{9}&\numprint{823343}&3.55&1.40&2.22&\bfseries 0.96&1.60&16.16&54.74&813.46\\
disk&\numprint{0}&\numprint{823611}&0.29&0.23&0.52&0.31&0.22&17.87&30.20&\bfseries 0.15\\
disk&\numprint{1}&\numprint{823519}&0.29&0.23&0.52&0.31&0.22&17.76&30.39&\bfseries 0.15\\
disk&\numprint{2}&\numprint{823129}&0.29&0.23&0.52&0.31&0.21&17.72&30.25&\bfseries 0.15\\
disk&\numprint{3}&\numprint{823260}&0.29&0.24&0.53&0.32&0.21&17.73&30.40&\bfseries 0.15\\
disk&\numprint{4}&\numprint{823734}&0.29&0.23&0.52&0.31&0.21&17.77&30.23&\bfseries 0.15\\
disk&\numprint{5}&\numprint{823556}&0.29&0.24&0.52&0.31&0.21&17.82&30.38&\bfseries 0.15\\
disk&\numprint{6}&\numprint{823506}&0.28&0.23&0.52&0.31&0.21&17.81&30.55&\bfseries 0.15\\
disk&\numprint{7}&\numprint{823590}&0.28&0.23&0.51&0.31&0.21&17.70&30.19&\bfseries 0.15\\
disk&\numprint{8}&\numprint{823442}&0.28&0.23&0.52&0.31&0.22&17.84&30.43&\bfseries 0.15\\
disk&\numprint{9}&\numprint{823791}&0.29&0.23&0.51&0.30&0.22&17.80&30.29&\bfseries 0.15\\

  \bottomrule
  \end{tabular}
    }
    \caption{Ratio 50\% insertions 50\% queries.}
    \label{tab:mixing:50:50:all}

\end{table}

\newpage

- 

\newpage

\begin{table}[H]
{%
    \begin{tabular}{lllS[table-format=2.2]S[table-format=2.2]S[table-format=2.2]S[table-format=2.2]S[table-format=2.2]S[table-format=2.2]S[table-format=2.2]S[table-format=4.2]}
    \toprule 
    &&\# True &&&\multicolumn{3}{c}{{{{\bucketNameBase}}}}\\\cmidrule(lr){6-8}
    &Seed&Queries&{{{\competitorAVL}}}&{{{\competitorBTree}}}&{{{\bucketLinearGrowSuffix}}}&{{{\bucketBtreeSuffix}}}&{{{\bucketHullSuffix}}}&{{{\competitorCHTree}}}&{{{\competitorCQTree}}}&{{{\competitorVector}}}\\
    \midrule
bell&\numprint{0}&\numprint{398307}&0.19&0.12&0.29&0.15&0.11&29.55&44.54&\bfseries 0.08\\
bell&\numprint{1}&\numprint{405326}&0.19&0.13&0.29&0.15&0.11&29.54&44.82&\bfseries 0.09\\
bell&\numprint{2}&\numprint{398217}&0.19&0.14&0.30&0.17&0.11&29.54&44.90&\bfseries 0.09\\
bell&\numprint{3}&\numprint{392900}&0.20&0.14&0.30&0.17&0.11&29.50&45.30&\bfseries 0.09\\
bell&\numprint{4}&\numprint{372110}&0.18&0.12&0.31&0.15&0.11&29.46&45.23&\bfseries 0.08\\
bell&\numprint{5}&\numprint{383752}&0.19&0.14&0.29&0.16&0.11&29.53&45.35&\bfseries 0.09\\
bell&\numprint{6}&\numprint{387694}&0.19&0.13&0.31&0.16&0.11&29.45&45.19&\bfseries 0.09\\
bell&\numprint{7}&\numprint{411733}&0.20&0.14&0.30&0.17&0.12&29.51&45.15&\bfseries 0.09\\
bell&\numprint{8}&\numprint{414112}&0.18&0.14&0.29&0.16&0.11&29.44&45.21&\bfseries 0.08\\
bell&\numprint{9}&\numprint{378721}&0.19&0.13&0.31&0.16&0.11&29.36&44.96&\bfseries 0.09\\
box&\numprint{0}&\numprint{524270}&0.22&0.16&0.13&0.18&0.12&30.46&47.54&\bfseries 0.10\\
box&\numprint{1}&\numprint{524281}&0.20&0.14&0.13&0.16&0.10&30.37&47.94&\bfseries 0.08\\
box&\numprint{2}&\numprint{524272}&0.20&0.14&0.13&0.16&0.11&30.22&47.95&\bfseries 0.09\\
box&\numprint{3}&\numprint{524275}&0.20&0.15&0.12&0.16&0.11&30.20&47.99&\bfseries 0.09\\
box&\numprint{4}&\numprint{524270}&0.22&0.17&0.13&0.19&0.12&30.24&47.81&\bfseries 0.09\\
box&\numprint{5}&\numprint{524274}&0.21&0.16&0.13&0.18&0.11&30.31&47.98&\bfseries 0.09\\
box&\numprint{6}&\numprint{524278}&0.21&0.15&0.13&0.17&0.11&30.42&48.13&\bfseries 0.09\\
box&\numprint{7}&\numprint{524267}&0.21&0.15&0.13&0.17&0.11&30.23&48.08&\bfseries 0.09\\
box&\numprint{8}&\numprint{524276}&0.21&0.15&0.12&0.16&0.11&30.28&48.04&\bfseries 0.08\\
box&\numprint{9}&\numprint{524279}&0.22&0.15&0.13&0.17&0.12&30.30&48.20&\bfseries 0.10\\
circle&\numprint{0}&\numprint{411435}&4.51&1.70&1.86&\bfseries 1.02&2.03&26.55&96.09&1831.32\\
circle&\numprint{1}&\numprint{411825}&4.82&1.79&1.86&\bfseries 1.02&2.03&26.92&92.88&1830.27\\
circle&\numprint{2}&\numprint{411096}&4.64&1.80&1.86&\bfseries 1.04&2.04&26.47&94.55&1830.16\\
circle&\numprint{3}&\numprint{412442}&4.65&1.68&1.85&\bfseries 1.01&2.02&27.13&94.86&1829.72\\
circle&\numprint{4}&\numprint{411446}&4.54&1.76&1.86&\bfseries 1.04&2.05&26.42&93.98&1832.46\\
circle&\numprint{5}&\numprint{411565}&4.66&1.75&1.84&\bfseries 1.01&2.02&26.24&94.90&1830.61\\
circle&\numprint{6}&\numprint{411744}&4.55&1.79&1.85&\bfseries 1.01&2.02&27.17&96.92&1830.67\\
circle&\numprint{7}&\numprint{411782}&4.68&1.75&1.85&\bfseries 1.03&2.04&26.40&95.58&1830.21\\
circle&\numprint{8}&\numprint{412035}&4.76&1.65&1.84&\bfseries 1.03&2.04&26.44&95.93&1832.21\\
circle&\numprint{9}&\numprint{411829}&4.68&1.76&1.87&\bfseries 1.03&2.04&26.80&95.38&1830.05\\
disk&\numprint{0}&\numprint{412153}&0.36&0.23&0.32&0.28&0.18&31.21&49.41&\bfseries 0.15\\
disk&\numprint{1}&\numprint{411312}&0.36&0.23&0.33&0.27&0.19&31.05&49.79&\bfseries 0.15\\
disk&\numprint{2}&\numprint{411847}&0.36&0.23&0.32&0.27&0.18&30.86&49.61&\bfseries 0.15\\
disk&\numprint{3}&\numprint{411547}&0.36&0.23&0.32&0.27&0.18&30.89&49.83&\bfseries 0.15\\
disk&\numprint{4}&\numprint{412199}&0.36&0.22&0.32&0.27&0.18&30.98&49.68&\bfseries 0.15\\
disk&\numprint{5}&\numprint{411501}&0.36&0.23&0.32&0.28&0.18&31.07&49.82&\bfseries 0.15\\
disk&\numprint{6}&\numprint{411963}&0.35&0.23&0.33&0.27&0.18&30.96&49.97&\bfseries 0.15\\
disk&\numprint{7}&\numprint{412072}&0.36&0.23&0.32&0.27&0.18&30.91&49.56&\bfseries 0.15\\
disk&\numprint{8}&\numprint{411629}&0.36&0.23&0.33&0.28&0.18&31.09&49.93&\bfseries 0.15\\
disk&\numprint{9}&\numprint{411977}&0.36&0.22&0.32&0.27&0.18&30.95&49.85&\bfseries 0.15\\

  \bottomrule
  \end{tabular}
    }
    \caption{Ratio: 75\% insertions 25\% queries.}
    \label{tab:mixing:75:25:all}
    
\end{table}

\clearpage

\section{\texorpdfstring{Data for scaling input size. \\}{Data for scaling input size.}}
\label{app:scaling}

In the previous section, we demonstrated that our algorithms exhibit different performance profiles across the four synthetic data sets. We now investigate how these performance profiles evolve as the input size increases.
For each of our seven algorithms and each of the four synthetic data sets, we define an experiment consisting of five runs. In each run, we fix the insertion-to-query ratio at 1:1 and vary the number of insertions (and thus queries) in the range $[2^{20}, 2^{26}]$. 
Table~\ref{tab:scaling} reports average runtimes (in seconds) for each experiment. Due to the large number of runs, we impose a timeout of 100 seconds per run.

\paragraph{Results discussion.}
Let $n$ denote the input size and $h = |CH^+(P)|$ the size of the convex hull.
The fully dynamic \competitorCHTree and \competitorCQTree are inefficient and quickly reach the timeout as $n$ increases.
All other algorithms, with the exception of \competitorVector, exhibit logarithmic scaling behaviour. The \competitorVector method also exhibits logarithmic scaling on the \textbf{Bell} and \textbf{Box} data sets. This is because its update time being linear in $h$, and $h$ grows logarithmically with $n$ on these data sets. On the \textbf{Disk} data set, the runtime increases more steeply; however, due to its superior cache locality, the method remains among the most efficient. In contrast, on the \textbf{Circle} data set (where $h \in \Theta(n)$) \competitorVector exceeds the time limit once $n > 2^{20}$.

\begin{table}[H]
\vspace{-1.5cm}
{%
    \begin{tabular}{llS[table-format=2.2]S[table-format=2.2]S[table-format=2.2]S[table-format=2.2]S[table-format=2.2]S[table-format=2.2]S[table-format=2.2]S[table-format=4.2]}
    \toprule 
    & &&&\multicolumn{3}{c}{{{{\bucketNameBase}}}}\\\cmidrule(lr){5-7}
    &$2^k$&{{{\competitorAVL}}}&{{{\competitorBTree}}}&{{{\bucketLinearGrowSuffix}}}&{{{\bucketBtreeSuffix}}}&{{{\bucketHullSuffix}}}&{{{\competitorCHTree}}}&{{{\competitorCQTree}}}&{{{\competitorVector}}}\\
    \midrule

{{{bell}}}&20&0.16&0.13&0.47&0.18&0.14&16.85&27.38&\bfseries 0.09\\
{{{bell}}}&22&0.65&0.55&2.26&0.74&0.55&110.06&149.25&\bfseries 0.36\\
{{{bell}}}&24&2.60&2.21&10.62&2.99&2.22&&&\bfseries 1.45\\
{{{bell}}}&26&10.57&9.03&48.86&12.14&8.88&&&\bfseries 5.79\\
{{{box}}}&20&0.17&0.15&0.14&0.18&0.13&17.29&29.23&\bfseries 0.09\\
{{{box}}}&22&0.71&0.62&0.57&0.74&0.53&113.07&159.07&\bfseries 0.37\\
{{{box}}}&24&2.87&2.48&2.23&2.96&2.11&&&\bfseries 1.49\\
{{{box}}}&26&11.59&9.94&9.01&11.95&8.50&&&\bfseries 5.96\\
{{{circle}}}&20&3.47&1.37&2.21&\bfseries 0.94&1.58&15.85&54.10&813.63\\
{{{circle}}}&22&22.68&10.62&13.30&\bfseries 5.51&8.07&102.31&370.18&\\
{{{circle}}}&24&126.95&64.20&74.26&\bfseries 28.41&38.74&&&\\
{{{circle}}}&26&678.06&347.77&405.01&\bfseries 139.45&180.79&&&\\
{{{disk}}}&20&0.29&0.23&0.52&0.31&0.21&17.68&30.27&\bfseries 0.15\\
{{{disk}}}&22&1.27&0.94&2.50&1.30&0.90&116.39&164.35&\bfseries 0.63\\
{{{disk}}}&24&5.68&3.96&11.94&5.79&3.92&&&\bfseries 2.69\\
{{{disk}}}&26&26.07&16.90&56.35&26.80&19.64&&&\bfseries 12.20\\

  \bottomrule
  \end{tabular}
    }
    \caption{Average Running time per data type for $2^k, k\in \{20,22,24,26\}$. Cells empty if previous run exceeded 100 seconds.}
    \label{tab:scaling}
\end{table}
\vspace{11cm}

The tree-based \competitorAVL and \competitorBTree algorithms have update and query time complexity $O(\log h)$. This behaviour is consistent with the experimental data: these methods scale well on the \textbf{Bell} and \textbf{Box} data sets, but their performance degrades when $h$ grows more rapidly with $n$.
In particular, their scaling on the \textbf{Bell} and \textbf{Box} data sets appears to be slightly better than other algorithms. It may be that there exists an input size where these techniques are the preferred method. 

Our logarithmic method algorithms demonstrate clean logarithmic scaling in $n$ across all data sets. This observation supports our claim that the theoretical $O(\log^2 n)$ query time bound is overly pessimistic in practice.
Among these, the \bucketLinearGrow method exhibits slightly inferior scaling on data sets where the convex hull is relatively large. Although the size of the bottom bucket is fixed, linear-time insertions into this bucket appear to introduce significant overhead.

\paragraph{Overall conclusion.}
For fixed input size, our algorithms exhibit distinct performance profiles. As the input size increases, their relative performance differences remain largely consistent. With the exception of \competitorVector on the \textbf{Disk} and \textbf{Circle} data sets, all algorithms demonstrate comparable scaling behaviour.

\clearpage
\section{\texorpdfstring{Memory Consumption.}{Memory consumption.}}\label{app:memory}

We briefly consider memory usage, focusing exclusively on the synthetic \textbf{Circle} data set, where every point lies on the convex hull. This scenario provides the fairest comparison for the fully dynamic \competitorCHTree and \competitorCQTree methods, which cannot discard points that do not appear on the hull.
In this experiment, we consider $n=2^{20}$  and $n=10^6$ insertions and no queries.
The full results, reporting peak memory consumption, are provided in Table~\ref{tab:memory_0} and \ref{tab:memory_0:1m}. Averaged outcomes for $n=10^6$ are plotted in  Figure~\ref{fig:mem:circle:1m}. For $n=2^{20}$ they are in the main paper.

\vspace{3.cm}

\begin{table}[H]
{%
    \begin{tabular}{llS[table-format=2.2]S[table-format=2.2]S[table-format=2.2]S[table-format=2.2]S[table-format=2.2]S[table-format=2.2]S[table-format=2.2]S[table-format=4.2]}
    \toprule 
    &&&&\multicolumn{3}{c}{{{{\bucketNameBase}}}}\\\cmidrule(lr){5-7}
    &Seed&{{{\competitorAVL}}}&{{{\competitorBTree}}}&{{{\bucketLinearGrowSuffix}}}&{{{\bucketBtreeSuffix}}}&{{{\bucketHullSuffix}}}&{{{\competitorCHTree}}}&{{{\competitorCQTree}}}&{{{\competitorVector}}}\\
    \midrule
circle&\numprint{0}&67.97&23.12&30.10&43.95&43.75&499.52&1065.83&\bfseries 19.94\\
circle&\numprint{1}&67.84&23.30&32.38&48.80&43.87&501.03&1067.17&\bfseries 19.90\\
circle&\numprint{2}&67.69&25.55&35.10&43.98&43.84&504.16&1063.95&\bfseries 19.87\\
circle&\numprint{3}&67.74&28.33&25.55&44.12&43.94&505.16&1061.02&\bfseries 19.90\\
circle&\numprint{4}&67.68&29.26&26.32&46.56&43.93&505.56&1063.19&\bfseries 19.97\\
circle&\numprint{5}&67.80&23.26&31.81&58.43&39.94&500.04&1072.21&\bfseries 19.90\\
circle&\numprint{6}&67.80&23.17&25.26&50.60&43.77&500.99&1064.59&\bfseries 20.72\\
circle&\numprint{7}&67.86&23.12&26.78&52.90&39.98&499.96&1065.94&\bfseries 19.99\\
circle&\numprint{8}&67.79&23.06&30.35&47.21&40.04&499.87&1067.50&\bfseries 20.07\\
circle&\numprint{9}&67.69&23.17&33.44&53.99&43.88&505.01&1064.76&\bfseries 19.82\\

  \bottomrule
  \end{tabular}
    }
    \caption{ Memory usage in MB after inserting $2^{20}$ points across 10 experiments on \textbf{Circle} data. }
    \label{tab:memory_0}
    
\end{table}

\begin{table}[H]
{%
    \begin{tabular}{llS[table-format=2.2]S[table-format=2.2]S[table-format=2.2]S[table-format=2.2]S[table-format=2.2]S[table-format=2.2]S[table-format=2.2]S[table-format=4.2]}
    \toprule 
    &&&&\multicolumn{3}{c}{{{{\bucketNameBase}}}}\\\cmidrule(lr){5-7}
    &Seed&{{{\competitorAVL}}}&{{{\competitorBTree}}}&{{{\bucketLinearGrowSuffix}}}&{{{\bucketBtreeSuffix}}}&{{{\bucketHullSuffix}}}&{{{\competitorCHTree}}}&{{{\competitorCQTree}}}&{{{\competitorVector}}}\\
    \midrule
circle&\numprint{0}&64.82&22.04&28.00&23.87&23.75&477.28&1014.49&\bfseries 19.06\\
circle&\numprint{1}&64.76&22.14&29.94&32.41&25.39&477.10&1013.30&\bfseries 19.06\\
circle&\numprint{2}&64.70&22.66&27.50&23.88&27.56&477.04&1016.32&\bfseries 19.00\\
circle&\numprint{3}&64.70&23.78&33.60&22.09&21.91&477.00&1015.85&\bfseries 19.10\\
circle&\numprint{4}&64.65&23.64&30.05&23.94&24.68&477.09&1015.97&\bfseries 19.06\\
circle&\numprint{5}&64.82&22.15&30.55&23.85&23.85&477.19&1016.99&\bfseries 19.03\\
circle&\numprint{6}&64.71&24.89&26.55&39.99&35.38&477.08&1017.61&\bfseries 19.10\\
circle&\numprint{7}&64.78&23.84&35.91&28.47&26.37&477.10&1015.26&\bfseries 19.16\\
circle&\numprint{8}&64.65&22.09&32.07&23.90&23.84&477.06&1016.85&\bfseries 19.14\\
circle&\numprint{9}&64.70&25.53&25.61&25.17&23.93&477.09&1014.19&\bfseries 18.95\\

  \bottomrule
  \end{tabular}
    }
    \caption{ Memory usage in MB after inserting $10^{6}$ points across 10 experiments on \textbf{Circle} data. }
    \label{tab:memory_0:1m}
    
\end{table}
\begin{figure}[H]
\vspace{-0.25cm}
    \centering
    \includegraphics[width=\linewidth]{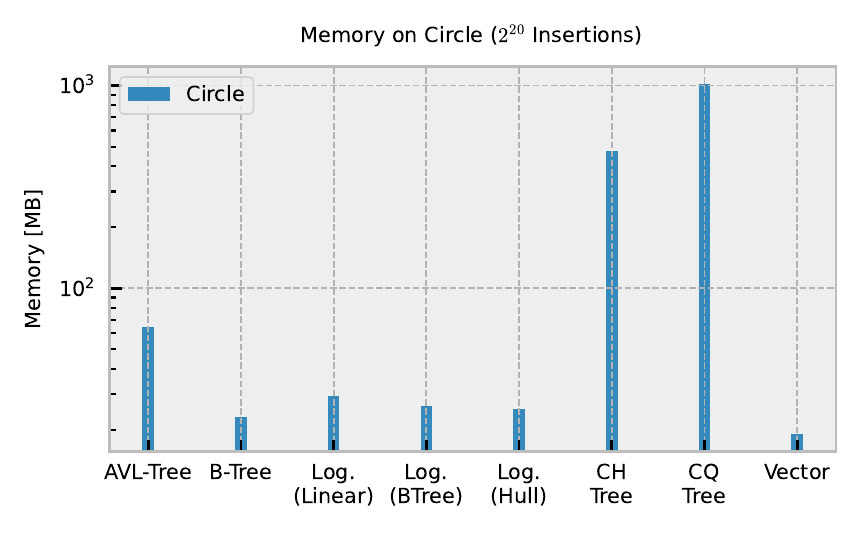}
    \caption{Memory consumption when $n  = 10^{6}$.}
    \label{fig:mem:circle:1m}
\end{figure}
\clearpage

\paragraph{Results discussion for $n=2^{20}$.}

Naturally, the \competitorVector algorithm has the smallest memory footprint, and we use it as the baseline for comparison.
Among the tree-based implementations, \competitorBTree uses approximately 1.23 times more memory than the baseline. In contrast, the AVL-based \competitorAVL consumes significantly more space, a factor 3.4, due to its reliance on numerous pointers and balance counters.
Our implementations based on the logarithmic method are also space-efficient. The \bucketLinearGrow method is the most compact among them, using only 1.49 times the memory of the baseline. The \bucketBtree and \bucketHull variants require 2.46 and 2.14 times more memory, respectively.
Finally, the fully dynamic \competitorCHTree and \competitorCQTree methods consume significantly more memory, using 25.20 and 53.50 times the baseline, respectively. Recall that on the \textbf{Circle} data set, all algorithms store the entire input point set. Hence, the observed overhead is attributable purely to the structure and metadata maintained by each data structure.
\paragraph{Results discussion for $n=10^{6}$.}
In Figure~\ref{fig:mem:circle:1m} and Table~\ref{tab:memory_0:1m} the results for only $n=10^6$ insertions are shown. the \competitorVector method requires the least memory and will be used as comparison again. All competitors (\competitorAVL,\competitorBTree, \competitorCHTree and \competitorCQTree) have the same relative additional footprint as above (within a small error margin). For our methods, we report different relative memory consumption: \bucketLinearGrow requires 1.57 times the memory. The methods growing by the actual hull size in the first bucket (\bucketBtreeSuffix and \bucketHullSuffix) require with a factor of 1.40 and 1.34 respectively much less memory than with $n=2^{20}$.
We can explain this behavior by the difference in the point when a merge is triggered and the fact that two separate half-hulles are maintained by the algorithms. \bucketLinearGrow expunges linearly for every 512 points seen and each separate hull contains 256 points only at the time.
This causes the second bucket to have also a size of 512 and then start the exponential growth.
The  other two \bucketNameBase methods (\bucketBtreeSuffix, \bucketHullSuffix) have to wait until 1024 points have passed. In each of the two separate hull only every second point is in the hull.
This leads to a difference in allocating the higher buckets.
For $n=10^6$ all buckets are nearly saturated in the \bucketBtreeSuffix and \bucketHullSuffix method.  
For $n=2^{20}$, which is the next bigger power of two, the two methods have to allocate a new bucket and move the points there, leaving the lower buckets empty, but still allocated. 

\paragraph{Conclusion.} All of our new methods are more efficient than the fully-dynamic competitors in a fair setting, where every point is in the convex hull. They also require less memory than \competitorAVL. Due to the logarithmic design, our algorithms may have empty, but allocated buckets impacting the memory footprint. Since allocating memory is expensive, we opt to keep them that way.

\section{\texorpdfstring{Further parameter studies. \\}{Further parameter studies.}}

In Section~\ref{sec:parameter} we considered the parameters of our algorithms for the logarithmic method. 
In particular, we considered the size of the bottom-most bucket $B_\ell$ and the data structure that maintains $CH^+(B_\ell)$ (via \competitorVector or \competitorBTree).
There is one parameter consideration that precedes that data structure, which is the parameter $B$ that our B-tree implementations use. 
In this section, we briefly verify what good choices for $B$ are.

\subsection{\texorpdfstring{Size of a node in the B-tree. \\}{Size of a node in the B-tree}}\label{app:parameter:btree}

The B-Tree implementation allows us to choose the size of a node in the B-tree in bytes. We tested a range of values $B\in\{16,256,1024,4096\}$, the results are shown in Table~\ref{tab:tuning:b:btree}.
Since $B=1024$ performs the best when averaging across both types of data, we choose it as default size of $B$ throughout the paper.

\begin{table}[H]
\begin{tabular}{lrrrrrr}\toprule
&\multicolumn{1}{c}{$B=16$}&\multicolumn{1}{c}{$B=256$}&\multicolumn{1}{c}{$B=1024$}&\multicolumn{1}{c}{$B=4096$}\\\midrule

Insertions&\\\midrule
box&0.10&0.08&0.07&\bfseries 0.07\\
circle&1.29&0.82&\bfseries 0.72&0.74\\
\midrule
Queries\\
\midrule
box&0.10&0.09&0.07&\bfseries 0.07\\
circle&1.16&0.71&0.60&\bfseries 0.60\\
\midrule
Overall\\\midrule
box&0.20&0.17&0.15&\bfseries 0.15\\
circle&2.45&1.53&\bfseries 1.32&1.34\\
\bottomrule  
\end{tabular}
\caption{ Our experiments for tuning the size of a node in bytes of the b-tree on $2^{20}$ insertions and queries. Running time is shown in in seconds, the fastest entry is \textbf{bold}.}
\label{tab:tuning:b:btree}
\end{table}

\clearpage
\section{\texorpdfstring{Implementing other queries. \\}{Implementing other queries.}}
\label{app:queries}

Chan identifies six convex hull queries~\cite{chan2011three}:
\begin{enumerate}[noitemsep]
    \item Finding the extreme point of $P$ in a given direction,
    \item Deciding whether a query line intersects $CH(P)$,
    \item Finding the hull vertices tangent to a query line,
    \item Deciding whether a point $q$ lies inside the area bounded by $CH(P)$,
    \item Finding the intersection points with a query line,
    \item Finding the intersection between two convex hulls.
\end{enumerate}

The first three queries are decomposable. That is, if one partitions $P$ into point sets $P_1, \ldots, P_k$ and maintains $CH^+(P_i)$ for each part then one may query these $k$ hulls separately to obtain the answer on $CH^+(P)$.

We describe in Section~\ref{sec:logarithmic} how to implement the fourth query for our \logarithmic and we use this query for our experiments. 
In this appendix, we show how to implement the fifth query also. 
We consider the implementation of the sixth query (convex hull intersection) to be out of scope for three reasons: first, its implementation is rather complex. Second, we are not aware of any existing implementation of this query (static or dynamic). Third, it is known in theory how to combine the predicates for the first five queries to obtain the sixth (e.g., Chazelle and Dobkin~\cite{Chazelle1980Detection}. We note that the actual implementation details involve a large number of case distinctions).

\subsection{\texorpdfstring{Implementing standard intersection. \\}{Implementing standard intersection. }}
Let $CH$ be a convex set of $h$ edges of an upper convex hull and let $\ell$ be a query line.
We first describe how to find the edge of $CH$ that is intersected by $\ell$ (if any such edge exists) using quadratic geometric predicates in $O(\log h)$ time. 
Our key observation is as follows. Let $(u, v)$ be an edge of $CH$ then:

\begin{itemize}[noitemsep]
    \item $\ell$ intersects $(u, v)$ if and only $u$ lies on one side of $\ell$ and $v$ on the other.  
    \item Otherwise if $u$ lies right of the intersection point $\ell \cap line(u, v)$ then $\ell$ must intersect an edge left of $(u, v)$ (if any)
    \item Otherwise $v$ lies left of the intersection point $\ell \cap line(u, v0$ and $\ell$ must intersect an edge right of $(u, v)$ (if any).
\end{itemize}

Consider the edges of $CH^+(P)$ in a balanced binary tree.
We use this observation to create a straightforward $O(\log n)$-time recursive querying algorithm to find the edge $(u, v)$ that intersects $\ell$ (our algorithm outputs $\emptyset$ if no such edge exists).

\begin{algorithm}[H]
    \caption{\texttt{LineIntersect}(edge $(u, v)$, line $\ell$)}
    \label{alg:lineintersect}
    \begin{algorithmic}
    \IF{$(u, v)$ is \texttt{null}}
    \RETURN $\emptyset$
    \ELSIF{\texttt{above\_line}($\ell$, $u$) AND $\neg$ \texttt{above\_line}($\ell$, $v$) }
    \RETURN $(u, v)$
    \ELSIF{\texttt{above\_line}($\ell$, $v$) AND $\neg$ \texttt{above\_line}($\ell$, $u$) }
     \RETURN $(u, v)$
    \ELSIF{\texttt{lies\_right}$(\ell, (u, v))$}
    \RETURN \texttt{LineIntersect}($(u, v)$.left, $\ell$)
    \ELSE
    \RETURN  \texttt{LineIntersect}($(u, v)$.right, $\ell$)
    \ENDIF
    \end{algorithmic}
\end{algorithm}

We showed in Section~\ref{sec:robustness} how to implement the helper function $\texttt{above\_line}(\ell, u)$. 
What remains is to implement \texttt{lies\_right}$(\ell, (u, v))$:

\begin{lemma}
    Let $\ell$ be a line and $(u, v)$ be a segment s.t. the supporting line has a different slope, then: 
    $\textnormal{\texttt{lies\_right}}(\ell, (u, v)) =$
    \begin{align*}
    &  \bigl( slope(\ell) < slope(line(u, v)) \textnormal{ AND } \textnormal{\texttt{above\_line}}( \ell, u)   \bigr) \\ 
    &\vee   \bigl( slope(\ell) > slope(line(u, v)) \textnormal{ AND } \neg \textnormal{\texttt{above\_line}}( \ell, u)   \bigr)
    \end{align*}
\end{lemma}

\begin{proof}
    Suppose that $slope(\ell) < slope(line(u, v))$. 
    Then $u$ lies right of $\ell \cap line(u, v)$ if and only if $u$ lies above the halfplane bounded from above by $\ell$.
    By symmetry, if $slope(\ell) > slope(line(u, v))$ then $u$ lies right of $\ell \cap line(u, v)$ if and only if $c$ lies below the halfplane bounded from above by $\ell$.
\end{proof}

\subsection{\texorpdfstring{Intersecting our logarithmic structure. \\}{Intersecting our logarithmic structure. }}
Finally, we show an algorithm to find the point of intersection between $CH^+(P)$ and a line $\ell$ in our \logarithmic.
Recall that this data structure partitions $P$ across buckets $B_i$. 

\begin{observation}
    If a line $\ell$ does not intersect $CH(B_i)$ for $i$ then $\ell$ does not intersect $CH^+(P)$.     
\end{observation}

\begin{definition}
    \label{def:intersect}
    Given our buckets $\{ B_i \}$ and a line $\ell$, denote by $U$ the set of all vertices such that there exists an index $i$ and an edge $(u, v)$ or $(v, u)$ of $CH^+(B_i)$ that intersects $\ell$ (each index $i$ may have two such edges). 
\end{definition}

\begin{definition}
    For any vertex $a$ and bucket $B_i$ denote by $L_i^a$ all edges of $B_i$ left of $a$ and by $R_i^a$ all edges right of $a$.
    For any vertex $a \in B_i$, denote by $bridges(a)$ the union over all $j$ of:
    \begin{itemize}[noitemsep]
        \item  the bridge between $L_i^a$ and $R_j^a$, and
        \item the bridge between $L_j^a$ and $R_i^a$.
    \end{itemize}
    For any set $U$ denote by $bridges(U) = \bigcup_{u \in U} bridges(u)$.    
\end{definition}

\begin{lemma}
   Consider our buckets $\{ B_i \}$, a line $\ell$, and the set of vertices $U$ from Definition~\ref{def:intersect}.
   If $\ell$ intersects an edge $(u, v)$ of $CH^+(P)$ then $(u, v) \in bridges(U)$.   
\end{lemma}

\begin{proof}
    Let $\ell$ intersect an edge $(u, v)$ of $CH^+(P)$ then there exist indices $i$ and $j$ such that $u$ is a vertex of $CH^+(B_i)$ and $v$ is a vertex of $CH^+(B_j)$. 
    If $i = j$ then the lemma immediately follows so let $i \neq j$. 

    Let $(a, b)$ and $(c, d)$ be the at most two edges of $CH^+(B_i)$ that is intersected by $\ell$. Similarly, let $(e, f)$ and $(g, h)$ be the at most two edges of $CH^+(B_j)$ that is intersected by $\ell$.
    If any of these edges has as a vertex $u$ or $v$, then the lemma immediately follows as $(u, v)$ must be an edge of $CH^+(B_i \cup B_j)$.
    
    If neither of these edges has as a vertex $u$ or $v$ then one of these vertices must lie  right of $u$ and left of $v$. 
    Denote by $x$ such a vertex.
    Then $L_i^x$ contains the vertex $u$ and $R_j^x$ contains the vertex $v$.
    It follows from $(u, v) \in CH^+(P)$ that $(u, v) \in CH^+(L_i^x \cup R_j^x)$.
    Any join of two convex hull that are separated by a vertical line has a unique edge that has one vertex left of the line and one vertex right of this line and we call this edge the bridge. Thus $(u, v)$ equals the bridge between $L_i^a$ and $R_j^a$ and so $(u, v) \in Bridges(U)$.  
\end{proof}

\end{document}